\documentclass[aoas]{imsart}

\RequirePackage{amsthm,amsmath,amsfonts,amssymb,dsfont,bbm,bm}
\RequirePackage[authoryear]{natbib}
\RequirePackage[colorlinks,citecolor=blue,urlcolor=blue]{hyperref}
\RequirePackage{graphicx}

\startlocaldefs

\usepackage[ruled]{algorithm2e} % Had to do before cleverref for reference to work 

\usepackage{cleveref} 
\crefname{proposition}{Proposition}{Propositions}
\Crefname{proposition}{Proposition}{Propositions}

\theoremstyle{plain}

\newtheorem{theorem}{Theorem}[section]
\newtheorem{proposition}[theorem]{Proposition}

\theoremstyle{remark}
\newtheorem{definition}[theorem]{Definition}

\newtheorem*{remark}{Remark}
\def\journal@name{} % This removes header saying "Submitted to..."

\usepackage{mathrsfs}
\newcommand{\V}{\mathcal{V}}

\newcommand{\X}{\mathcal{X}}
\newcommand{\Xb}{\mathscr{X}}

\newcommand{\M}{\mathcal{M}}
\newcommand{\C}{\mathcal{C}}

\newcommand{\setx}{X}
\newcommand{\sety}{Y}
\newcommand{\setz}{Z}
\newcommand{\seqx}{X}
\newcommand{\seqy}{Y}
\newcommand{\seqz}{Z}
\newcommand{\elmtx}{x}
\newcommand{\elmty}{y}
\newcommand{\elmtz}{z}

\usepackage{subcaption}
\endlocaldefs

\begin{document}

\begin{frontmatter}

\title{Distances for Comparing Multisets and Sequences}
%\title{A sample article title with some additional note\thanksref{t1}}
\runtitle{Distances for Comparing Multisets and Sequences}
%\thankstext{T1}{A sample additional note to the title.}

\begin{aug}
%%%%%%%%%%%%%%%%%%%%%%%%%%%%%%%%%%%%%%%%%%%%%%%
%% Only one address is permitted per author. %%
%% Only division, organization and e-mail is %%
%% included in the address.                  %%
%% Additional information can be included in %%
\author[A]{\fnms{George}~\snm{Bolt}\ead[label=e1]{g.bolt@lancaster.ac.uk}},
\author[B]{\fnms{Sim\'{o}n}~\snm{Lunag\'{o}mez}\ead[label=e3]{simon.lunagomez@itam.mx}}
\and
\author[A]{\fnms{Christopher}~\snm{Nemeth}\ead[label=e2]{c.nemeth@lancaster.ac.uk}\orcid{0000-0002-9084-3866}}
%%%%%%%%%%%%%%%%%%%%%%%%%%%%%%%%%%%%%%%%%%%%%%
%% Addresses                                %%
%%%%%%%%%%%%%%%%%%%%%%%%%%%%%%%%%%%%%%%%%%%%%%
\address[A]{Lancaster University\printead[presep={,\ }]{e1,e2}}

\address[B]{Instituto Tecnológico Autónomo de México (ITAM)\printead[presep={,\ }]{e3}}
\end{aug}

\begin{abstract}
Measuring the distance between data points is fundamental to many statistical techniques, such as dimension reduction or clustering algorithms. However, improvements in data collection technologies has led to a growing versatility of structured data for which standard distance measures are inapplicable. 
% Often one must instead aggregate the data to another form, such as a vector, before considering distances thereof.
% Often quantification of distances in such instances requires some form of aggregation, converting observations to a standard structure, such as a vector, before invoking known distances. However, when one aggregates a loss of information occurs, and this may perhaps want to be avoided. 
In this paper, we consider the problem of measuring the distance between sequences and multisets of points lying within a metric space, motivated by the analysis of an in-play football data set. Drawing on the wider literature, including that of time series analysis and optimal transport, we discuss various distances which are available in such an instance. For each distance, we state and prove theoretical properties, proposing possible extensions where they fail. Finally, via an example analysis of the in-play football data, we illustrate the usefulness of these distances in practice.
\end{abstract}

% \begin{keyword}
% \kwd{distance}
% \kwd{multiset}
% \kwd{distance}
% \kwd{multiset}
% \end{keyword}

\end{frontmatter}
%%%%%%%%%%%%%%%%%%%%%%%%%%%%%%%%%%%%%%%%%%%%%%
%% Please use \tableofcontents for articles %%
%% with 50 pages and more                   %%
%%%%%%%%%%%%%%%%%%%%%%%%%%%%%%%%%%%%%%%%%%%%%%
%\tableofcontents

\section{Introduction}
\label{sec:intro}
Distance measures represent a versatile tool for the practicing data analyst. Once a distance has been specified, an array of subsequent methodologies immediately become available. These include clustering algorithms such as hierarchical clustering \cite[Sec. 12.3]{izenman2008modern} and DBSCAN \citep{ester1996density}, placing data points into groups; dimension reduction or embedding techniques, such as multidimensional scaling (MDS) \cite[Ch. 13]{izenman2008modern} and UMAP \citep{mcinnes2018umap,becht2019dimensionality}, facilitating data visualisation; or prediction algorithms such as k-nearest neighbours regression \cite[Sec. 13.3]{hastie2009elements}. 
% Distances can also provide the foundation for specification of probability distributions over non-standard objects, such as graphs \citep{lunagomez2021modeling} or ranks \citep{Vitelli2018a}.

However, with improvements in data collection comes an increasingly diverse array of structured data for which standard distance measures are unsuitable. This motivates consideration of distances tailored to fit the objects of focus. Examples include graph distances \citep{Donnat2018b}, appearing frequently in the network data analysis literature, or distances between ranks \citep{kumar2010generalized}, which often appear in the context of preference learning.

In this paper, we consider the problem of eliciting distances between sequences and multisets. In the most general sense, a sequence is an enumerated collection of objects within some underlying space, with a multiset being the un-ordered analogue of a sequence. An intuitive example is a text document. Naturally, this can be seen as a sequence of words. However, it can also be seen as a multiset of words, or what is referred to by some as a `bag-of-words' \citep{kusner2015word}, wherein two documents equal up to a permutation of word order would be considered one and the same. Other examples of data interpretable in this manner are
\begin{itemize}
    \item Temporal networks, for example, in the analysis of \cite{Donnat2018b} biological measurements were encoded via a graph for a given patient through a longitudinal study;
    \item User interactions within online platforms. For example, the Foursquare data set \citep{Yang2015} records users checking into different venues throughout the day, e.g. cinemas, cafes, sports venues etc., which leads to a sequence of sequences for each user, with the inner sequence for a user representing one day of their venue check-ins;
    \item Historical purchases, where the purchase history of a single customer could be represented as a sequence or multiset of orders, with each order encoded as set of products. Such data often appears in the  market basket analysis literature \citep{raeder2011market}.
\end{itemize}

Another instance of data in this form, and one that will serve as our running example, can be obtained from an in-play football data set shared by StatsBomb.\footnote{\url{https://github.com/statsbomb/open-data}} This incredibly rich data set contains high-granularity information concerning events within football matches, and of particular interest to us is information regarding passes. In particular, it is possible from these data to infer, for a given team in a given match, series of un-interrupted passes between their players. Intuitively, each series of passes can be seen as a path over player positions (\Cref{fig:football_ex}). Moreover, over the course of a football match this will lead to a series of such paths being observed, for example, \Cref{fig:football_ex} shows the first ten enacted by England in a match against Italy during the UEFA Euro 2020 competition. As such, a football match (for a given team) can be seen as a sequence or multiset of paths.

% It will be assumed throughout that one has access to a notion of distance between the objects within the underlying space. Indeed, in most cases this is likely to be necessary. Given the underlying space is of a reasonable size it is unlikely for any two observations to share much in common. Nonetheless, some observations are likely to be `closer' than others, and to have any hope of quantifying this one must have some notion of distance in the underlying space. 

Though the problem of eliciting distances between sequences and multisets is not new, little has been done in the way of an overarching review. Moreover, oftentimes these have been addressed separately, with no recognition of the connections inherent from the fact sequences and multisets are closely related. Herein lies the motivation for this work. The intention is for this to serve as a point for reference for anyone faced with data of this structure; which we feel is a very general one. The only restriction we impose is that a distance metric be defined over the underlying space. For each distance, we provide an intuitive interpretation, prove theoretical properties and discuss how they can be computed. 

The remainder of this paper will be structured as follows. In \Cref{sec:background}, we introduce the notation to be used throughout and provided background on distance metrics. \Cref{sec:multisets} then introduces distances to compare multisets, whilst \Cref{sec:sequences} does so for sequences. Finally, in \Cref{sec:data_analysis} we illustrate the use of these distances in practice through an analysis of the StatsBomb data set, where we consider visualising data structure via a dimension reduction technique.

\begin{figure}
    \centering
    \includegraphics[width=0.95\linewidth]{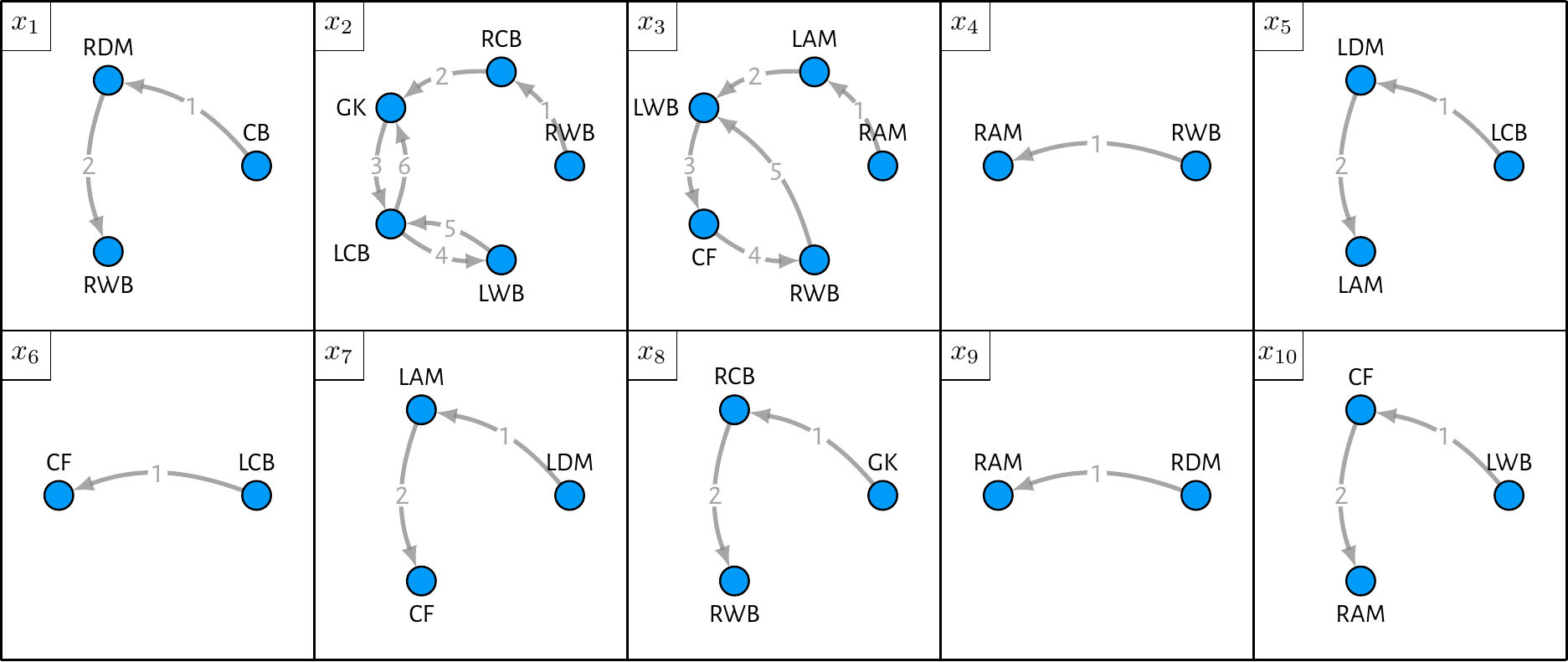}
	\caption{Data from a single match of the StatsBomb data set, where each $\elmtx_i$ represents a un-interrupted series of passes by the given team, with indices indicating order of observation, that is, $\elmtx_i$ was observed before $\elmtx_{i+1}$. Here vertices correspond to player positions (abbreviated according to \Cref{table:position_abbrv}) and edges represent passes, labelled according to the order in which they occurred. These particular observations correspond to the first ten series of passes made by England in their game against Italy during the UEFA Euro 2020 competition. }
	\label{fig:football_ex}
\end{figure}

\section{Background and notation}

\label{sec:background}

A single observation will be denoted by $X$, which may represent either a sequence or multiset. A sequence we denote as follows
$$\seqx=(\elmtx_1,\dots,\elmtx_N)$$
where $x_i \in \X$ for some general space $\X$. For example, regarding the football data (\Cref{fig:football_ex}), $\X$ would denote the space of all paths over the player positions. A multiset is the order-invariant analogue of a sequence and we denote it as follows
$$\setx = \{\elmtx_1,\dots, \elmtx_N\}$$
where $x_i\in \X$, with the curly braces $\{\}$ being used to signify this is a multiset and hence the order of elements therin is arbitrary. Note in both we allow $\elmtx_i = \elmtx_j$ for $i\not=j$, and hence the need to opt for multiset over regular sets. We let $|\setx|=N$ denote sequence length, or equivalently multiset cardinality. 

A multiset $X$ can also be represented via a function $m_\setx \,: \, \X \to \mathbb{Z}_+$ where $m_\setx(\elmtx)$ denotes the multiplicity of $x$ in $\setx$, which we refer to as the \textit{multiplicity function}. Moreover, this defines the support of $\setx$ in $\X$ as follows $$\text{Supp}(X) := \{x \in \X \,:\, m_X(x) > 0\},$$ denoting the set of unique elements in $\setx$. As an example, we might have $\X = \mathbb{Z}_+$ with $X=\{1,1,1,2,2,3\}$ a multiset, where $m_\setx(1)=3$, $m_\setx(2)=2$ and $m_\setx(3)=1$, whilst $\text{Supp}(X)=\{1,2,3\}$.

A distance measure over the space $\Xb$ is a function $d \, : \, \Xb \times \Xb \to \mathbb{R}_+$, taking as input two elements of the space and outputting some measure of dissimilarity between them. It is natural to require that such functions satisfy certain properties, which are formalised mathematically via the notion of a distance metric.

\begin{definition}[Distance metric]
    \label{def:metric}
    A function $d \, : \, \Xb \times \Xb \to \mathbb{R}_+$ is a \textit{distance metric} over the space $\Xb$ if, for any $x,y,z \in \Xb$, the following conditions are satisfied
    \begin{longlist}
% 	\item $d(x,y) \geq 0$  (non-negativity);\label{metric_condition1:non_neg}
	\item $d(x, y) = 0 \iff x = y$  (identity of indiscernibles);\label{metric_condition2:indentity}
	\item $d(x,y) = d(y,x)$  (symmetry);\label{metric_condition3:symmetry}
	\item $d(x,y) \leq d(x,z) + d(z,y)$ (triangle inequality);\label{metric_condition4:tri}
    \end{longlist}
    with the pair $(\Xb,d)$ being referred to as a \textit{metric space}.
\end{definition}

Of particular interest in this work are distance measures between sequences and multisets, so that $\Xb$ of \Cref{def:metric} would denote the space of all sequences or multisets over the underlying space $\X$.
% Introducing the following notion for the space of all sequences (using elements of $\X$) 
% $$\mathbf{S}(\X) = \{(x_1,\dots,x_N) \, : \, x_i \in \X,\, N \geq 1\} $$
% and for the analagous space of multisets
% $$\mathbf{M}(\X) = \{\{x_1,\dots,x_N\} \, : \, x_i \in \X,\, N \geq 1 \}$$
% where here we abuse notation slightly, with the outer parenthesis denoting a set whilst the inner denote a multiset.
% This implies we seek to define functions $d$ where either \begin{align*}
%     d \, : \, \seqspc(\X) \times \seqspc(\X) \to \mathbb{R}_+ && \text{or} && d \, : \, \setspc(\X) \times \setspc(\X) \to \mathbb{R}_+.
% \end{align*}
As mentioned in the introduction, towards defining such distances it will be assumed that one has access to a distance $d \, : \, \X \times \X \to \mathbb{R}_+$ over the underlying space $\X$, which we refer to as the \textit{ground distance}. For example, regarding the football data, this would amount to a distance between paths. Moreover, it will be assumed that $d$ satisfies the conditions of \Cref{def:metric} (with $\Xb=\X$), and hence is a distance \textit{metric}. In this way, the multiset or sequence $\seqx$ can be seen as collections of points within the metric space $(\X,d)$.

% For clarity, we summarise this in the following assumption.

% \begin{assumption}
%   A distance metric $d$ over the underlying space $\X$ is defined.
% \end{assumption}

Finally, we discuss distance normalisation. Often when comparing objects of different sizes via distance measures it can be useful to normalise them. A solution is to use an approach adopted by \cite{Donnat2018b}, based on a metric transformation. This transform is referred to therein as the \textit{Steinhaus transform}, but is also seen in \cite{Deza2009} where it is referred to as the \textit{biotope transform}. Given a distance metric $d$ (note it must be a metric) over the space $\Xb$, with $c \in \Xb$ some reference element of this space, the Steinhaus transform of $d$ is given by
\begin{align}
	\label{eq:steinhaus}
	\bar{d}(x,y) := \frac{2d(x,y)}{d(x,c) + d(y,c) + d(x,y)},
\end{align}
defining a new distance $\bar{d}$, which can be shown to be a metric. Note, we leave out any reference to $c$ in this notation, though one should be aware that by definition $\bar{d}$ does depend on it. Observe that since $d$ is a metric, and hence obeys the triangle inequality \eqref{metric_condition4:tri}, we have the following result
\begin{equation*}
    \begin{aligned}
      d(x,y) &\leq d(x,c)+d(c,y)  \\
     \implies 2d(x,y) &\leq d(x,c)+d(y,c)+d(x,y) \\
     \implies \bar{d}(x,y) &\leq 1
    \end{aligned}
\end{equation*}
that is, $\bar{d}$ is bounded. Moreover, it will to be non-negative since it is a ratio of non-negative terms. As such we have $\bar{d}(x,y) \in [0,1]$ for any $x,y \in \Xb$.

\section{Distances between multisets}
\label{sec:multisets}
In this section, we outline distance measures one can use to compare multisets. All of the distances here share a similar structure, each considering possible relations between elements of either observation, seeking to find the relation which is in some sense `optimal'. The objective to optimise is typically some notion of cost, and it is the minimal value of this cost which is taken as the distance. Where these measures differ is in the structure of this relation. For each distance in this section, we give an intuitive and formal definition, discuss theoretical properties and provide details on how they are computed. 

% \subsection{Hausdorff Distance}

\subsection{Matching distances}
\label{sec:matching_distances}

A natural route to defining a distance between multisets is to consider pairing the elements from multiset with the other, defined formally via the notion of a \textit{matching} (\Cref{fig:ex_matching}). For each matching, one can make use of the ground distance between set elements to define a notion of cost. A distance is then defined by finding the minimum cost matching. The resulting distance also has an alternative interpretation; seen as the minimum cost of turning one multiset into another by (i) inserting or deleting elements with some specified cost or, (ii) substituting one element for another at a cost proportional to their dissimilarity.

Formally, given two multisets $\setx$ and $\sety$ a \textit{matching} is a multiset of pairs 
\begin{equation}
    \begin{aligned}
        \M = \{(\elmtx, \elmty) \, : \, \elmtx \in \setx, \, \elmty \in \sety\}
    \end{aligned}
    \label{eq:matching}
\end{equation}
such that each $\elmtx\in \setx$ is matched to at most one $\elmty \in \sety$, taking into account multiplicities, and \textit{vice versa}. Equivalently, each $\elmtx \in \text{Supp}(\setx)$ can be matched to at most $m_\setx(\elmtx)$ elements $y \in \text{Supp}(\sety)$, and \textit{vice versa}. Observe by this definition that one must have $0 \leq |\M| \leq \min(|\setx|,|\sety|)$, and a matching which achieves this upper bound is said to be \textit{complete}. For example, the matching of \Cref{fig:ex_matching} is complete. Finally, we define
$$\M_\setx := \{x\in \setx \, : \, \exists\,\, \elmty \in \sety, \text{ with }  (\elmtx,\elmty) \in \M\}$$
so that $\M_\setx \subseteq \setx$ denotes the elements of $\setx$ which are included in the matching $\M$, whilst we introduce the shorthand $\M_\setx^c := \setx \setminus \M_\setx$ to denote the elements of $\setx$ \textit{not} included in the matching $\M$.

For any given matching $\M$ we can assign it a cost as follows 
% \begin{equation*}
% 	C(\M) = \underset{\text{Matched elements}}{\underbrace{\sum_{(\elmtx,\elmty) \in \M} d(\elmtx,\elmty)}} + \underset{\text{Penalty un-matched elements}}{\underbrace{\lambda(\M)}}
% 	\label{eq:gen_matching_cost}
% \end{equation*}
\begin{equation}
	C(\M) = \sum_{(\elmtx,\elmty) \in \M} d(\elmtx,\elmty) + \lambda(\M)
	\label{eq:gen_matching_cost}
\end{equation}
with $d(\cdot,\cdot)$ is the ground distance over the underlying space $\X$ and $\lambda(\M) \geq 0$ is some penalty term for un-matched elements, that is, we sum the pairwise distances of matched elements and penalise un-matched elements. A distance between $\setx$ and $\sety$ is now defined by minimising this cost over all matchings. 

Each choice for $\lambda(\M)$ will define a different distance, and we consider two. For the first, we use the ground distance $d(\cdot,\cdot)$, letting
\begin{equation*}
    \lambda(\M) = \sum_{\elmtx \in  \M_\setx^c} d(\elmtx, \Lambda) + \sum_{\elmty \in \M_\sety^c} d(\elmty, \Lambda)
\end{equation*}
where $\Lambda \in \X$ denotes a reference value, typically taken to be the null or equivalent, with $d(\elmtx,\Lambda)$ often capturing some notion of size for the element $\elmtx \in \X$, though this will depend on the choice of metric $d$ and the underlying space. For example, if $x \in \X$ are paths we might take $\Lambda$ to be the empty path.

\begin{definition}[Matching distance]
    For two multisets $\setx$ and $\sety$ the matching distance is given by the following
	\begin{align}
		d_{\mathrm{M}}(\setx,\sety) := \min_{\M} \left\{  \left(  \sum_{(\elmtx,\elmty) \in \M} d(\elmtx,\elmty) \right) + \sum_{\elmtx \in \M_\setx^c} d(\elmtx, \Lambda) + \sum_{\elmty \in \M_\sety^c} d(\elmty, \Lambda)\right\}
		\label{dist:matching}
	\end{align}
	where $\M$ denotes a matching of $\setx$ and $\sety$, and $\Lambda \in \X$ denotes a reference element of $\X$, typically the null element.
	\label{def:matching}
\end{definition}

An alternative approach is to penalise each un-matched entry equally by some pre-specified amount $\rho>0$, that is 
\begin{equation*}
    \begin{aligned}
        \lambda(\M) &=\rho \times \text{(\# un-matched elements)}\\
        &=\rho( |\setx| + |\sety| - 2|\M| ),
    \end{aligned}
\end{equation*}
which leads to the following distance.

\begin{definition}[Fixed-penalty matching distance]
    For two multisets $\setx$ and $\sety$ the fixed-penalty matching distance is given by the following
	\begin{equation}
		\begin{aligned}
			d_{\mathrm{M}, \rho}(\setx,\sety) := \min_{\M} \left\{ \sum_{(\elmtx,\elmty) \in \M} d(\elmtx,\elmty) + \rho( |\setx| + |\sety| - 2|\M| ) \right\}
			\label{dist:fp_matching}
		\end{aligned}
	\end{equation}
	where $\M$ is a matching of $\setx$ and $\sety$, and $\rho>0$ is a parameter controlling the penalty for un-matched elements.
	\label{def:fp_matching}
\end{definition}

\begin{figure}
    \centering
    \begin{subfigure}{0.48\linewidth}
        \centering
        \vspace{0.2cm}
        \includegraphics[width=0.5\linewidth, page=1]{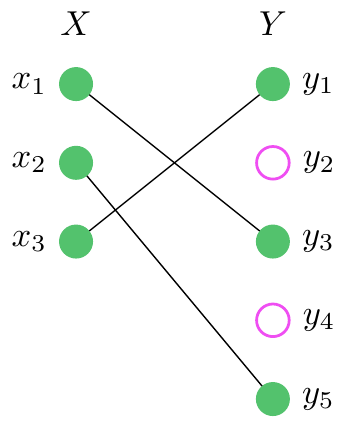}   
        \vspace{0.2cm}
        \caption{Matching of $\setx$ and $\sety$}
        \label{fig:ex_matching}
    \end{subfigure}
    \begin{subfigure}{0.48\linewidth}
        \centering
        \vspace{0.2cm}
        \includegraphics[width=0.5\linewidth, page=2]{img/example_relations.pdf}   
        \vspace{0.2cm}
        \caption{Coupling of $\mu_\setx$ and $\mu_\sety$}
        \label{fig:ex_coupling_emd}
    \end{subfigure}
    \caption{Example relations found when evaluating multiset distances, with (a) showing a matching $\M$ of the multisets $\setx = \{\elmtx_1,\dots,\elmtx_3\}$ and $\sety= \{\elmty_1, \dots, \elmty_5\}$, whilst (b) shows a coupling $\mathbf{P}$ of the distributions $\mu_\setx$ and $\mu_\sety$, where the edge from $\elmtx_i$ to $\elmty_j$ is proportional to $\mathbf{P}_{ij}$, the mass moved from $\elmtx_i \in \X$ to $\elmty_j \in \X$, and node (circle) radii of $\elmtx_i$ and $\elmty_i$ are proportional to $\mu_\setx(\elmtx_i)$ and $\mu_\sety(\elmty_j)$, respectively. For simplicity, here we assume the elements of $\setx$ and $\sety$ are distinct, so that within $\mu_\setx$ and $\mu_\sety$ the masses are equal.}
\end{figure}

Computation of these distances requires finding an optimal matching, which can be achieved via the Hungarian algorithm \citep{Kuhn1955}. This is an algorithm proposed to solve the assignment problem, which seeks an optimal assignment of $n$ `workers' to $n$ `tasks', doing so with a complexity of $\mathcal{O}(n^4)$. As we detail in \Cref{sec:computing_matching_distances}, by setting up the right optimisation problem, we can obtain these two distances at a computational complexity of $\mathcal{O}(\max(N,M)^4 + N M)$, where $N=|\setx|$ and $M=|\sety|$, with the $\max(N,M)^4$ term due to the Hungarian algorithm, whilst the $N M$ term arises through the need to evaluate all pairwise distances between elements of $\setx$ and $\sety$.

We now discuss some theoretical properties of both distances, proofs of which can be found in \Cref{proof:matching_is_metric}. Firstly, both are distance metrics (\Cref{def:metric}).

\begin{proposition}
    Both $d_{\mathrm{M}}$ and $d_{\mathrm{M},\rho}$ satisfy metric conditions \eqref{metric_condition2:indentity}-\eqref{metric_condition4:tri}.
	\label{thm:matching_is_metric}
	\label{thm:fp_matching_is_metric}
\end{proposition}

We also have results regarding the form of optimal matchings for each distance, which are particularly useful when it comes to evaluating these distances via the Hungarian algorithm. 

\begin{proposition}
    \label{thm:matching_complete}
    For the matching distance $d_{\mathrm{M}}$, there always exists a complete matching achieving the optimum of \cref{dist:matching}.
\end{proposition}

\begin{proposition} 
    \label{thm:fp_matching_complete}
    For the fixed-penalty matching distance $d_{\mathrm{M},\rho}$ with $$\rho \geq \frac{1}{2} \left[\max_{\elmtx\in \setx, \elmty\in\sety} d(\elmtx,\elmty)\right]$$ there exists a complete matching achieving the optimum of \cref{dist:fp_matching}.
\end{proposition}

As consequence of \Cref{thm:matching_complete,thm:fp_matching_complete}, one only needs to optimise over complete matchings when the relevant conditions hold. As such, the size of optimisation problem to be solved via the Hungarian algorithm can be minimised (\Cref{sec:computing_matching_distances}). 

\Cref{thm:fp_matching_complete} also sheds light on how one might choose $\rho$. Following the rationale that when an optimal matching is incomplete one is to some extent ignoring information, it makes sense to choose $\rho$ such that a complete optimal matching can always be found. From \Cref{thm:fp_matching_complete}, one can see this will depend on the pairwise distances of $\setx$ and $\sety$. However, if the ground distance $d$ happens to be a bounded, that is $d(x,y) \leq K$ for all $x,y \in \X$ and $0 < K < \infty$, then by \Cref{thm:fp_matching_complete} if $\rho \geq K/2$ one is guaranteed to have a complete optimal matching. Moreover, if one would like as much of the distance to be driven by the pairwise distances as possible then $\rho = K/2$ is a sensible choice. Though having a bounded distance appears restrictive, recall that via the Steinhaus transform of \cref{eq:steinhaus} one can obtain a bounded distance with $K=1$ given \textit{any} distance metric. 

We finish by noting that both matching distances are similar to those proposed by others. In particular, \cite{Ramon2001} and \cite{eiter1997distance} both considered the problem of comparing sets within a metric space, though they considered genuine sets whereas we consider multisets, with both defining their respective measures via an optimal relationship between the two sets. Of the two, \cite{Ramon2001} is most similar, and in fact the vernacular and notation of matchings that we adopted here was inspired by theirs.

\subsection{Earth mover's distance}
\label{sec:emd}
Though theoretically sound, a drawback of the matching distances is that when $\setx$ and $\sety$ are of different sizes the pairwise information of certain elements is to an extent ignored, with the contribution of un-matched elements coming solely via the penalisation terms. Towards defining a distance which avoids such issues, one can use ideas from the literature on Optimal Transport (OT) \citep{Peyre2019}, an area of research which has considered the problem of quantifying the dissimilarity of probability distributions over general metric spaces. Namely, by converting multisets to distributions an OT-based distance thereof can be serve as a proxy for a distance between the original observations.

% Optimal Transport is an area of research focussed on quantifying the dissimilarity of probability distributions \citep{Peyre2019}. The key idea is to find the minimal `cost' of transporting the mass from one distribution to the other, where we imagine shovelling the mass from distribution to the other and assign the cost as the amount of work this involves, naturally making use of some measure of distance over the underlying space. These ideas have been round for a while, with foundations in linear programming and optimisation, but have recently become popular in areas such as machine learning and computer vision (couple of references), where the induced distance is sometimes referred to as the \textit{earth mover's distance}.  

We convert a multiset $\setx$ to a distribution as follows. Define $\mu_\setx \, : \, \X \to [0,1]$ via
\begin{equation}
    \begin{aligned}
        \mu_\setx(x) := \frac{m_\setx(x)}{|\setx|},
    \end{aligned}
    \label{eq:multiset_distribution_def}
\end{equation}
with $\mu_\setx(x)$ seen as the probability mass located at $x \in \X$.
Given two multisets $\setx$ and $\sety$, we now consider using an OT-based distance between $\mu_\setx$ and $\mu_\sety$ to measure their dissimilarity; namely the 1-Wasserstein distance \citep[Prop. 2.2]{Peyre2019}, also known as the \textit{earth mover's distance} (EMD).
% Given two multisets $X$ and $Y$ with $\text{Supp}(X)=\{x_1,\dots,x_K\}$ and $\text{Supp}(Y)=\{y_1,\dots,y_L\}$, one can convert these to distributions as follows
% \begin{align}
%     \label{eq:multisets_to_distributions}
%     \mu_\setx = \{(\elmtx_1, p_1),\dots, (\elmtx_{K}, p_{K})\} && \mu_\sety = \{(\elmty_1, p'_1),\dots, (\elmty_{L}, p'_{L})\}
% \end{align}
% where 
% \begin{align*}
%     p_i = \frac{m_\setx(\elmtx_i)}{N} && p'_i = \frac{m_\sety(\elmty_i)}{M}
% \end{align*}
% so that $\mathbf{p}=(p_1,\dots,p_K)$ and $\mathbf{p}'=(p_1',\dots,p_L')$ denote probability vectors. By viewing $\mu_\setx$ and $\mu_\sety$ as discrete distributions over the metric space $(\X,d)$, one can now consider using an OT-based distance to measure their dissimilarity. 

The EMD admits the following intuition. One imagines that $\mu_\setx$ and $\mu_\sety$ represent quantities of mass at various locations within the space $\X$, that is, there is $\mu_\setx(\elmtx)$ mass at $\elmtx\in\X$ and $\mu_\sety(\elmty)$ mass at $\elmty\in\X$. Moreover, one considers transforming $\mu_\setx$ into $\mu_\sety$ by `transporting' the mass from one set of locations to the other. Assuming the cost of moving mass between two points is proportional to their distance, that is, moving one unit of mass from $x\in \X$ to $y \in \X$ incurs a cost $d(x,y)$, the EMD is then defined to be the minimum cost required to transform $\mu_\setx$ into $\mu_\sety$. 

Formally, this can be cast as a linear optimisation problem. Note that by definition, any $\mu_\setx$ and $\mu_\sety$ have non-zero mass at a finite number of points in the space $\X$, namely at $\text{Supp}(X)=\{x_1,\dots,x_K\}$ and $\text{Supp}(Y)=\{y_1,\dots,y_L\}$, respectively. As such, in transforming $\mu_\setx \to \mu_\sety$ we need only consider the movement of mass between this finite collection of locations. The decision variables will now be the mass sent from $x_i \in \text{Supp}(\setx)$ to $y_j \in \text{Supp}(\sety)$ for each pair $(\elmtx_i,\elmty_j)$, which we denote $\mathbf{P}_{ij}$ and collate into the $K \times L$ matrix $\mathbf{P}$ (\Cref{fig:ex_coupling_emd}). Furthermore, with $\mathbf{D}$ the $K \times L$ matrix of pairwise distances, where $\mathbf{D}_{ij} = d(x_i, y_j)$, the goal is to find a $\mathbf{P}$ minimising the total cost, that is
\begin{equation*}
    \begin{aligned}
        \min \sum_{i=1}^K \sum_{j=1}^L \mathbf{D}_{ij}\mathbf{P}_{ij}
    \end{aligned}
\end{equation*}
subject to the constraints 
\begin{equation*}
    \begin{aligned}
        \sum_{j=1}^L \mathbf{P}_{ij} = \mu_\setx(\elmtx_i) \quad\text{(for $i=1,\dots,K$)} && \text{and} &&
        \sum_{i=1}^K \mathbf{P}_{ij} = \mu_\sety(\elmty_j) \quad\text{(for $j=1,\dots,L$)} 
    \end{aligned}
\end{equation*}
which ensure that $\mathbf{P}$ defines a movement of mass which starts with the distribution $\mu_\setx$ and ends with $\mu_\sety$, as desired. The EMD is subsequently defined to be the total cost of an optimal $\mathbf{P}$.

Towards a more succinct definition, we adopt notation of \cite{Peyre2019}, denoting a set of feasible $\mathbf{P}$ as follows
$$\mathbf{U}(\mu_\setx,\mu_\sety) := \{\mathbf{P} \in \mathbb{R}_+ \, : \mathbf{P} \mathbbm{1}_L = \mathbf{p}_\setx, \mathbf{P}^\mathrm{T} \mathbbm{1}_K = \mathbf{p}_\sety\}$$
where $\mathbbm{1}_N = (1,\dots,1)$ is the length $N$ vector of ones, and 
\begin{align*}
        \mathbf{p}_\setx= (\mu_\setx(\elmtx_1),\dots,\mu_\setx(\elmtx_K)) &&  \mathbf{p}_\sety= (\mu_\sety(\elmty_1),\dots,\mu_\sety(\elmty_L))
\end{align*}
denote probability vectors associated with each distribution. Adopting the vernacular therein, we refer to any $\mathbf{P}\in \mathbf{U}(\mu_\setx,\mu_\sety)$ as a \textit{coupling} of $\mu_\setx$ and $\mu_\sety$. With this, the EMD can be defined a follows.

%\begin{figure}
%	\centering
%	\includegraphics[width=0.6\linewidth]{img/tikz/EMD.pdf}
%	\caption{Illustrating distance metrics. }
%\end{figure}

% Let $d_{ij} = d(\elmtx_i, \elmty_j)$ denote the distance between $\elmtx_i$ and $\elmty_j$, sometimes referred to as the \textit{ground distance}. This will represent the cost of sending one unit of mass from $\elmtx_i$ to $\elmty_j$. Now the task is to find a flow $\bm{F} = [f_{ij}]$, with $f_{ij}$ the flow between $\elmtx_i$ and $\elmty$, that minimises the overall cost 
% \begin{align}
% 	\label{eq:EMD_lp_1}
% 	\min \sum_{i=1}^n \sum_{j=1}^m f_{ij}d_{ij}
% \end{align} 
% subject to the constraints
% \begin{equation}
% 	\begin{aligned}
% 		\label{eq:EMD_lp_2}
% 		&f_{ij} \geq 0, 1\leq i \leq n, \, 1 \leq j \leq m \\
% 		&\sum_{j=1}^m f_{ij} \leq w_{\elmtx_i}, \, 1 \leq i \leq n \\
% 		&\sum_{i=1}^n f_{ij} \leq w_{\elmty_j}, \, 1 \leq j \leq m \\
% 		&\sum_{i=1}^m \sum_{j=1}^m f_{ij} = \min\left\{\sum_{i=1}^n w_{\elmtx_i}, \sum_{j=1}^m w_{\elmty_j}\right\}
% 	\end{aligned}
% \end{equation}
% Assuming one can solve this linear program, the EMD between $\seqx$ and $\seqy$ can be defined:
\begin{definition}[Earth mover's distance]
    For two multisets $\setx$ and $\sety$ the earth movers distance is given by the following
    \begin{align}
		d_{\mathrm{EMD}}(\setx, \sety) &:=\min_{\mathbf{P} \in \mathbf{U}(\mu_\setx,\mu_\sety)} \sum_{i=1}^K\sum_{j=1}^L\mathbf{D}_{ij} \mathbf{P}_{ij}
	\end{align}
	where $\mu_\setx$ and $\mu_\sety$ are the distributions obtained from $\setx$ and $\sety$ as defined by eq. \eqref{eq:multiset_distribution_def}.
	\label{def:emd}
\end{definition}

Computation of the EMD reduces to solving a linear optimisation problem; specifically the transportation problem. As such, one can appeal to literature on solvers thereof \cite[details can be found in][Ch. 3]{Peyre2019}. There also exist packages in various programming languages which can be used to implement these algorithms easily, for example, the Python Optimal Transport (POT) toolbox \citep{flamary2021pot}.

We now consider theoretical properties of $d_{\mathrm{EMD}}$ as a distance between multisets. Since the EMD is a distance metric between probability distributions \cite[Prop. 2.2]{Peyre2019}, some properties will be naturally inherited. However, thanks to the normalisation enacted when constructing distributions via \cref{eq:multiset_distribution_def}, not all of the metric conditions will hold, as summarised by the following result (proof in \Cref{proof:emd_metric_cond}).

\begin{proposition}
    The earth mover's distance $d_{\mathrm{EMD}}$ satisfies metric conditions \eqref{metric_condition3:symmetry} (symmetry) and \eqref{metric_condition4:tri} (triangle inequality), but fails \eqref{metric_condition2:indentity} (identity of indiscernibles). 
    \label{thm:emd_metric_cond}
\end{proposition}

The failure of condition \eqref{metric_condition2:indentity} occurs when one multiset is a multiple of the other, that is, if there is some $C>0$ such that $m_\setx(x)=C \cdot m_\sety(x)$ for all $x \in \X$. However, assuming the multisets $\setx$ and $\sety$, and the underlying space $\X$, are all of reasonable size, the chances of this occurring are likely to be low. As such, though \Cref{thm:emd_metric_cond} may appear unattractive, the practical consequences are unlikely to be severe; though this will clearly depend on how one intends to use the distance. In any case, if necessary, one can extend $d_{\mathrm{EMD}}$ to define a valid metric as follows.

\begin{definition}[Earth mover's distance with cardinality comparison]
    \begin{align}
		d_{\mathrm{sEMD}}(\setx, \sety) := \tau d_{\mathrm{EMD}}(\setx, \sety)+ (1-\tau) d_s(|\setx|, |\sety|)
	\end{align}
	where $d_s(\cdot, \cdot)$ denotes a distance metric between integer values, whilst $0 < \tau < 1$ controls the relative contributions of $d_{\mathrm{EMD}}$ and $d_s$ to the overall distance.
\end{definition}

\begin{proposition}
    The distance $d_{\mathrm{sEMD}}$ satisfies metric conditions \eqref{metric_condition2:indentity}-\eqref{metric_condition4:tri}.
    \label{thm:ex_emd_is_metric}
\end{proposition}

Again, we note this approach to compare multisets via the EMD is not a new idea. For example, \cite{kusner2015word} did exactly this to define a distance between text documents.

%It also appears to uncover heterogeneity in a sample of Elsevier users (\Cref{fig:embedding_users}), though further analysis would be needed to confirm this is not simply due to noise in the data. 
%Further work is to consider whether a genuine distance  metric can be defined which takes a similar approach to the EMD, a starting place is to consider the work of \cite{Gardner2014}.

% \begin{figure}
% 	\centering
% 	\begin{subfigure}{0.48\linewidth}
% 		\centering
% 		\includegraphics[width=0.8\linewidth]{img/tikz/EMD.pdf}
% 	\end{subfigure}
% 	\begin{subfigure}{0.48\linewidth}
% 		\centering
% 		\includegraphics[width=0.8\linewidth]{img/tikz/GED.pdf}
% 	\end{subfigure}
% 	\caption{Illustrating distance metrics between interaction sequences. \textit{Left} regards the EMD, showing two sets of points within the metric space $(\Ib, d)$, that is, each point represents an interaction. The area of each point is proportional to their mass (multiplicity), with dashed arrows indicating a possible flow $\bm{F}$, moving the mass from one set of points to the other. \textit{Right} regards the GED, showing a valid matching $\mathcal{M}$ between two sequences of interactions. }
% 	\label{fig:EMD_GED}
% \end{figure} 

\section{Distances between sequences}
\label{sec:sequences}

In this section, we turn to the problem of measuring the dissimilarity of sequences, introducing two distances taken from the time series literature. These are typically interpreted as some form of minimum cost transformation, but can also be defined via an optimal relation between the two observations, much like the multiset distances. Again, for each distance we give an intuitive and formal definition before discussing theoretical and computational aspects.

\subsection{Edit distances}

\label{sec:edit_distance}
% As in \Cref{sec:matching_distances}, a natural way to define a distance between sequences is to consider pairing the elements of either observation. By assigning a cost to a pairing, one can subsequently define a distance by minimising the cost over all pairings. 

Here we introduce what we call the \textit{edit distance}, closely related to the Geometric Edit Distance (GED) \citep{Gold2018,Fox2019}. As the name suggests, this can be seen as the minimum cost of transforming one sequence into the other by (i) substituting one entry for another at a cost proportional to their dissimilarity, and (ii) inserting and deleting entries with some pre-specified penalty. It is also closely related to the matching distances (\Cref{sec:matching_distances}), admitting an alternative interpretation via an optimal matching between the two sequences; though the matching must in this case satisfy extra conditions to ensure the preservation of order. It is via this latter interpretation that we provide a formal definition.

Suppose that $\seqx = (\elmtx_1, \dots, \elmtx_N)$ and $\seqy = (\elmty_1, \dots, \elmty_M)$ are the two sequences to be compared. Observe the notion of a matching as stated in \cref{eq:matching} continues to make sense for sequences. However, since entries have an ordering one can further constrain the form of this matching. Namely, following \cite{Gold2018}, a matching $\M$ of $\setx$ and $\sety$ is said to be \textit{monotone} if for any $(\elmtx_{i_1}, \elmty_{j_1}),(\elmtx_{i_2}, \elmty_{j_2}) \in \M$ we have
 $$i_1 < i_2 \iff j_1<j_2$$ 
which ensures that $\M$ preserves the ordering of each sequence, or more informally and visually, when one draws the matching no lines cross (\Cref{fig:ex_monotone_matching}).

\begin{figure}
    \centering
    \begin{subfigure}{0.48\linewidth}
        \centering
        \vspace{0.2cm}
        \includegraphics[width=0.5\linewidth, page=3]{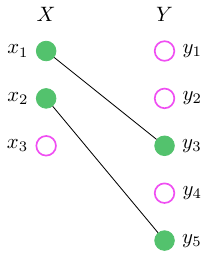} 
        \vspace{0.2cm}
        \caption{Montone matching}
        \label{fig:ex_monotone_matching}
    \end{subfigure}
    \begin{subfigure}{0.48\linewidth}
        \centering
        \vspace{0.2cm}
        \includegraphics[width=0.5\linewidth, page=4]{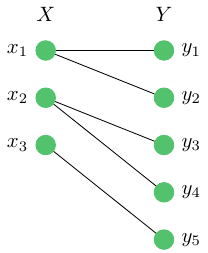}    
        \vspace{0.2cm}
        \caption{Sequence coupling}
        \label{fig:ex_seq_coupling}
    \end{subfigure}
    \caption{Example relations used to define sequence distances, where (a) shows an example of a monotone matching of the two sequences $\seqx$ and $\seqy$, used to define the edit distances, whilst (b) shows a coupling, used to define the dynamic time warping distances.}
\end{figure}

Following the approach taken for the matching distances (\Cref{sec:matching_distances}), by assigning a cost to each matching, a distance can be defined by minimising the cost over all feasible matchings. Our choice of cost functions for sequences will have similar features as the cost functions for multisets, where we (i) sum pairwise distances of matched entries, and (ii) penalise un-matched entries. Each choice of penalty again defines a different distance, and we consider using the exact same penalties as in \Cref{sec:matching_distances}, defining the edit distance and fixed-penalty edit distance as follows.
 
 \begin{definition}[Edit distance]
 For two sequences $X$ and $Y$, the edit distance is given by the following
 \begin{align}
		d_{\mathrm{E}}(\setx,\sety) := \min_{\M} \left\{  \left(  \sum_{(\elmtx,\elmty) \in \M} d(\elmtx,\elmty) \right) + \sum_{\elmtx \in \M_\setx^c} d(\elmtx, \Lambda) + \sum_{\elmty \in \M_\sety^c} d(\elmty, \Lambda)\right\}
		\label{dist:edit}
	\end{align}
	where $\M$ denotes a monotone matching of $\setx$ and $\sety$, and $\Lambda \in \X$ denotes a reference element of $\X$, typically the null element.
	\label{def:edit}
 \end{definition}
 
 \begin{definition}[Fixed-penalty edit distance]
    For two seqeunces $\setx$ and $\sety$ the fixed-penalty edit distance is given by the following
	\begin{equation}
		\begin{aligned}
			d_{\mathrm{E}, \rho}(\setx,\sety) := \min_{\M} \left\{ \sum_{(\elmtx,\elmty) \in \M} d(\elmtx,\elmty) + \rho( |\setx| + |\sety| - 2|\M| ) \right\}
			\label{dist:fp_edit}
		\end{aligned}
	\end{equation}
	where $\M$ is a monotone matching of $\setx$ and $\sety$, and $\rho>0$ is a parameter controlling the penalty for un-matched entries.
	\label{def:fp_edit}
 \end{definition}

Notice \Cref{def:edit,def:fp_edit} are more-or-less identical to \Cref{def:matching,def:fp_matching}; the difference being monotonicity of matchings. As with the matching distances, both can be shown to satisfy all metric conditions, summarised via the following result (proofs in \Cref{sec:proofs_sequences}).

\begin{proposition}
Both $d_\mathrm{E}$ and $d_{\mathrm{E},\rho}$ satisfies metric conditions \eqref{metric_condition2:indentity}-\eqref{metric_condition4:tri}
\label{thm:edit_distance_metric}
\end{proposition}

Regarding computation, both distances can be evaluated via dynamic programming at a complexity $\mathcal{O}(|\setx| \cdot |\sety|)$, further details of which can be found in \Cref{sec:computing_edit_distance}.

% \begin{remark}
% 	It should be noted this has similarities to LCS distance seen in \Cref{sec:dist_interactions}, which itself can be interpreted in terms of a minimum cost transformation via insertions and deletions. In fact both the GED and LCS are special cases of the more general Edit Distance seen in \cite{Wagner1974}, originally proposed therein for strings. This means that the dynamic programming algorithm used to compute the LCS distance can be similarly applied for the GED, with the only required alteration being a change of edit costs. 
% \end{remark}

%{\color{red} \textit{To do: how to choose $\rho$? (connection with matching distance)}} 
%\textit{Suppose that $0 \leq d \leq K$ is bounded, such that $d(\elmtx, \elmty) = K$ is interpreted as $\elmtx$ and $\elmty$ being nothing alike.
%The first thing to note is that it might make sense to constrain $\rho \geq K/2$. This can be seen by comparing the cost of inserting and deleting entries with substitution of one entry for another. Specifically, it makes sense to have 
%$$ \underset{\text{substitution cost}}{d(\elmtx, \elmty)} \leq \underset{\text{delete-insert cost}}{ \rho + \rho}$$
%that is, the cost to substitute $\elmtx$ for $\elmty$ should always be less than or equal to the cost of deleting $\elmtx$ then inserting $\elmty$. Now if $d$ is bounded above by $K$ then we can ensure this is the case by taking 
%$$2 \rho \geq K$$
%where solving leads to the desired constraint. The next thing to examine is what happens as we increase $\rho$ above this lower bound of $K/2$. 
%}

\subsection{Dynamic time warping}

Though the edit distances come with the theoretical benefits of being metrics, when faced with observations of differing lengths, much like the matching distances, they only really take into account pairwise information of matched entries, effectively ignoring un-matched ones.  This similarly motivates the need for a distance without such a feature. Interestingly, an answer can be found with another distance often seen in the time series literature. Namely, the \textit{dynamic time warping} (DTW) distance \citep{Gold2018}. 

In defining the DTW distance, we will use the notation and vernacular of \cite{Gold2018}. Like the edit distance, the DTW distance is based upon finding a minimum cost relation between the two sequences. The key difference between the two is the form of this relation; where the edit distance considered a monotone matching, the DTW considers a \textit{coupling} of the two sequences (\Cref{fig:ex_seq_coupling}). Note this sequence-based coupling differs from the coupling of distributions used to define the EMD (\Cref{sec:emd}). 
 Given two sequences $\seqx$ and $\seqy$, a coupling is a sequence of pairs
 $\C = (p_1,\dots, p_R)$, where each $p_r = (\elmtx_i,\elmty_j)$ for some with $1 \leq i \leq N$ and $1 \leq j \leq M$. To be a coupling, $\C$ must have the first and last entries paired together, that is $p_1 =(\elmtx_1,\elmty_1)$ and $p_R=(\elmtx_N,\elmty_M)$, and must satisfy the following
\begin{equation*}
    \begin{aligned}
        p_r = (x_i,y_j)  \implies p_{r+1} \in \{(\elmtx_i,\elmty_{j+1}), (\elmtx_{i+1},\elmty_j),(\elmtx_{i+1},\elmty_{j+1})\},
    \end{aligned}
\end{equation*}
that is, given $\elmtx_i$ and $\elmty_j$ are paired, one can either (i) pair the next two entries $\elmtx_{i+1}$ and $\elmty_{j+1}$, or (ii) enact some \textit{warping}, where an entry from either sequence is paired with more than one from the other. For example, in \Cref{fig:ex_seq_coupling} we see warping for the first and second entries of $\seqx$. Notice that by definition every entry of one sequence will always be coupled with at least one entry from the other.
% In the context of time series, where $\seqx$ and $\seqy$ are sequences of real values and indices correspond to time, this has the effect of slowing-down or speeding-up (warping) time for one sequence relative to the other, hence the name. 

To define a distance, one now assigns each coupling $\C$ a cost by summing the pairwise distances of coupled entries before minimising this cost over all couplings, leading to the following.

\begin{definition}
   Given sequences $\seqx$ and $\seqy$, the dynamic time warping distance is given by the following
   \begin{equation}
        \label{eq:dtw}
        \begin{aligned}
            d_{\mathrm{DTW}}(\seqx, \seqy) := \min_{\C}\left\{\sum_{(\elmtx,\elmty) \in \C} d(\elmtx, \elmty) \right\}
        \end{aligned}
    \end{equation}
    where $\C$ is a coupling. 
    \label{def:dtw}
\end{definition}

It should be noted the DTW distance has certain theoretical shortcomings. Specifically, it violates the identity of indiscernibles (\ref{metric_condition2:indentity}) and the triangle inequality (\ref{metric_condition4:tri}). This we summarise with the following result, a proof of which can be found in \Cref{sec:metric_proofs}.

\begin{proposition}
    \label{thm:dtw_metric_cond_violate}
    The dynamic time warping distance $d_{\mathrm{DTW}}$ satisfies metric condition \eqref{metric_condition3:symmetry} (symmetry), but violates conditions \eqref{metric_condition2:indentity} (identity of indiscernibles) and \eqref{metric_condition4:tri} (triangle inequality).
\end{proposition}

 Depending on the desired application, this may or may not be a significant issue. In the former case, it can be helpful to consider whether one can ensure satisfaction of at least one of these conditions. This motivates the following extension, obtained by inclusion of a warping penalty.
 
\begin{definition}
   Given sequences $\seqx$ and $\seqy$, the fixed-penalty DTW distance is given by the following
   \begin{equation}
        \label{eq:fp_dtw}
        \begin{aligned}
            d_{\mathrm{DTW}, \rho}(\seqx, \seqy) := \min_{\C}\left\{\sum_{(\elmtx,\elmty) \in \C} d(\elmtx, \elmty)  + \rho \cdot w(\C)\right\}
        \end{aligned}
    \end{equation}
    where
    $$w(\C) :=  |\{(\elmtx_i,\elmty_j) \in \C \, : \, (\elmtx_i,\elmty_{j+1}) \in \C \text{ or } (\elmtx_{i+1},\elmty_j) \in \C \}|, $$
    quantifies the amount of warping in $\C$, whilst $\rho > 0$ is a parameter controlling the penalisation incurred for each instance of warping.
    \label{def:fp_dtw}
\end{definition}

As a result of introducing this warping penalty the distance now satisfies the identity of indiscernibles \eqref{metric_condition2:indentity}, as summarised in the following result. 

\begin{proposition}
    The fixed penalty dynamic time warping distance $d_{\mathrm{DTW},\rho}$ satisfies metric conditions \eqref{metric_condition2:indentity} (identity of indiscernibles) and \eqref{metric_condition3:symmetry} (symmetry), but violates \eqref{metric_condition4:tri} (triangle inequality).
    \label{thm:fp_dtw_metric_conds}
\end{proposition}

Regarding computation, both DTW distances can be evaluated via dynamic programming at a time complexity of $\mathcal{O}(|\seqx| \cdot |\seqy|)$, with further details found in \Cref{sec:computing_dtw_distance}.

\section{Data analysis: Embedding football matches}
\label{sec:data_analysis}

Returning to the in-play football data, we now show how the distances of \Cref{sec:multisets,sec:sequences} can be used to visualise the structure present therein. In particular, given a choice of distance, we use MDS to obtain a two-dimensional representation of the data, often referred to as an embedding, which can then be plotted. 

% Furthermore, we consider comparing embeddings via different distances.

With the StatsBomb data processed into paths (\Cref{fig:football_ex}), we are left with a sample
$$\setx^{(1)},\dots,\setx^{(n)}$$
where each $\setx^{(i)}$ represents all pass sequences enacted by a particular team in a single match. Note each will lead to two sequences or multisets (one for each team), and for this data set we have 1096 games, leading to $n=2192$ observations. Depending on whether one would like to take order into account, these can be represented as sequences or multisets, that is
\begin{align*}
    \seqx^{(i)} = \left(x_1^{(i)},\dots,x_{N^{(i)}}^{(i)}\right) && \text{or} &&\seqx^{(i)} = \left\{x_1^{(i)},\dots,x_{N^{(i)}}^{(i)}\right\} 
\end{align*}
where $x_j^{(i)}$ denotes the $j$th path appearing in the $i$th observation, an in \Cref{fig:football_ex}

For a given distance between multisets or sequences, MDS outputs an embedding of data points into $m$-dimensional Euclidean space such that the pairwise distances between data points are best preserved. More specifically, each data point $\seqx^{(i)}$ gets associated a vector $\mathbf{x}_i \in \mathbb{R}^m$ such that $||\mathbf{x}_i-\mathbf{x}_j||_2 \approx d(\seqx^{(i)}, \seqx^{(j)})$ for each pair $(i,j)$, where $||\cdot||_2$ denotes the Euclidean norm. Though the dimension $m$ is general, typically we take $m=2$ so that each $\mathbf{x}_i$ can be plotted in 2-dimensional space, thus providing a visual summary of the structure present in the observed sample (with respect to the chosen distance). 

% Firstly, one constructs the distance matrix $\mathbf{D} \in \mathbb{R}_+^{n\times n}$ where $\mathbf{D}_{ij} = d(\seqx^{(i)},\seqx^{(j)})$ for each pair of data points. Now, one seeks an embedding of these data points in $m$-dimensional Euclidean space such that the pairwise distances between data points in preserved. More specifically, to each data point $\seqx^{(i)}$ we associate a vector $\mathbf{x}_i \in \mathbb{R}^m$, such that $||\mathbf{x}_i-\mathbf{x}_j||_2 \approx \mathbf{D}_{ij}$ for each pair of data points,  where $||\cdot||_2$ denotes the Euclidean norm. Though the dimension $m$ is general, typically we take $m=2$ so that each $\mathbf{x}_i$ can be plotted in 2-dimensional space.

In this analysis, we compare embeddings obtained in this manner for four different distances, two for sequences and multisets respectively. In particular, we consider the following
\begin{itemize}
  \item Sequence distances 
  \begin{enumerate}
    \item $\bar{d}_{\mathrm{E},\rho}$ : Fixed-penalty edit distance (\Cref{def:edit}), normalised via the Steinhaus transform of \cref{eq:steinhaus};
    \item $d_{\mathrm{DTW}}$ : Dynamic time warping distance (\Cref{def:dtw});
  \end{enumerate}
  \item Multiset distances 
  \begin{enumerate}
      \item $\bar{d}_{\mathrm{M},\rho}$ : Fixed-penalty matching distance (\Cref{def:fp_matching}), normalised via the Steinhaus transform of \cref{eq:steinhaus};
      \item $d_{\mathrm{EMD}}$ : Earth mover's distance (\Cref{def:emd}).
  \end{enumerate}
\end{itemize}

We must also choose our ground distance, which in this case amounts to specifying a distance metric between paths. A natural approach here is to consider finding maximally-sized common substructures, such as subpaths and subsequences (\Cref{fig:path_distances}). Suppose $x=(x_1,\dots,x_n)$ and $y=(y_1,\dots,y_m)$ are two paths, where $x_i, y_j \in \V$ for some set of vertices $\V$, for example, for the football data we have $\V$ denoting the set of player positions. One can now define the longest common subsequence (LCS) distance between $x$ and $y$ as follows
 $$d_{\text{LCS}}(x, y) := n + m - 2\delta_{\text{LCS}}$$
 where $\delta_{\mathrm{LCS}}$ is the maximum length of any subsequence shared by both $x$ and $y$. Intuitively, this can be seen as the number of entries of either path \textit{not} included in this maximum common subsequence (underlined entries in \Cref{fig:path_distances_subseq}), or equivalently the minimum number of entries one must delete and insert to transform one path into other. For example, the paths in \Cref{fig:path_distances_subseq} would have a LCS distance of 5. Further details regarding this distance (and its subpath analogue), including proofs of metric conditions and details regarding computation, can be found in \Cref{sec:path_distances}. 
 
The data analyst is free to choose the penalisation parameter $\rho>0$ in $\bar{d}_{\mathrm{E},\rho}$ and $\bar{d}_{\mathrm{M},\rho}$. Here we also consider normalising the ground distance via the Steinhaus transform eq. \eqref{eq:steinhaus}, leading to a ground distance of $\bar{d}_{\mathrm{LCS}}$, so that by the rationale discussed in \Cref{sec:matching_distances} we take $\rho=0.5$ (since $\bar{d}_{\mathrm{LCS}}(x,y) \leq 1$). 

\begin{figure}
    \centering
    \begin{subfigure}{0.48\linewidth}
        \centering
        \includegraphics[width=0.6\textwidth,page=2]{img/distances_paths.pdf}
        \caption{Common subpath}
        \label{fig:path_distances_subpath}
    \end{subfigure}
    \begin{subfigure}{0.48\linewidth}
        \centering
        \includegraphics[width=0.6\textwidth,page=1]{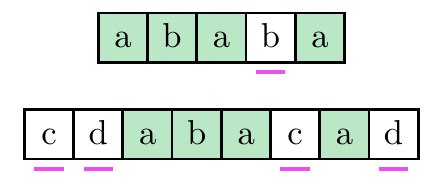}
        \caption{Common subsequence}
        \label{fig:path_distances_subseq}
    \end{subfigure}
    \caption{A comparison of common subsequences and subpaths. In (a) and (b) we see the same pair of paths, with (a) highlighting a common subpath, as indicated by shaded (green) entries, whilst (b) shows a common subsequence. In both cases, these are in fact maximal.}
    \label{fig:path_distances}
\end{figure}

\Cref{fig:football_embeddings} visualises the embeddings obtained for each of these distances. Note, to simplify the comparison the embeddings have been aligned via rotation and reflection, allowable since pairwise Euclidean distances are preserved under such transformations. Here we observe a strong similarity in structure across \Cref{fig:football_embedding_edit,fig:football_embedding_match,fig:football_embedding_emd}, with each showing clear clusters of data points. In contrast, in \Cref{fig:football_embedding_dtw} we see an embedding which is qualitatively different from the others. To highlight what might be driving the structure observed across these embeddings, data points have also been labelled according to the formation a team was playing most often, so that a data point labelled "433" implies the given team spent most of the time in the corresponding match with a formation consisting of 4 defenders, 3 midfielders and 3 attackers. Here one can observe in \Cref{fig:football_embedding_edit,fig:football_embedding_match,fig:football_embedding_emd,fig:football_embedding_dtw} that positions in the respective embedded space appear to be congruent with the formation a team was playing, so that two data points which are near one another therein are likely to be using a similar formation. Moreover, for \Cref{fig:football_embedding_edit,fig:football_embedding_match,fig:football_embedding_emd} there is a strong correspondence between clusters and formations, with the "433" formation being a good example, though some formations, such as "4231" and "442", appear to have more than one cluster. 
% This feels somewhat intuitive. Most significantly, the formation will determine the positions being used, making  whilst more subtly, the formation may also determine the flow of passes. 

\begin{figure}
    \centering
    \begin{subfigure}{0.48\linewidth}
	    \centering
	    \includegraphics[width=0.9\linewidth]{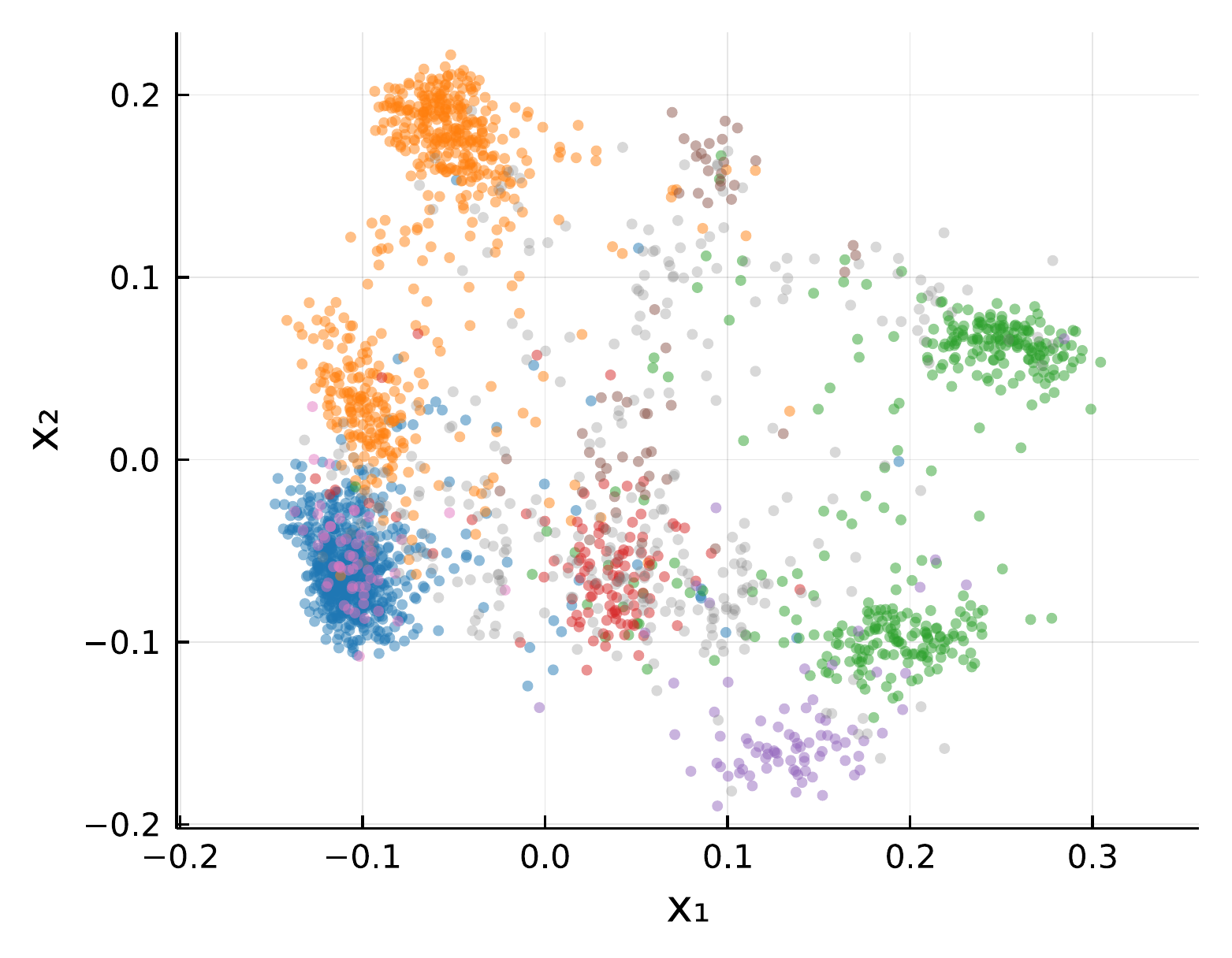}
	    \caption{$\bar{d}_{\mathrm{E},0.5}$}
	    \label{fig:football_embedding_edit}
	\end{subfigure}
	\begin{subfigure}{0.48\linewidth}
	    \centering
	    \includegraphics[width=0.9\linewidth]{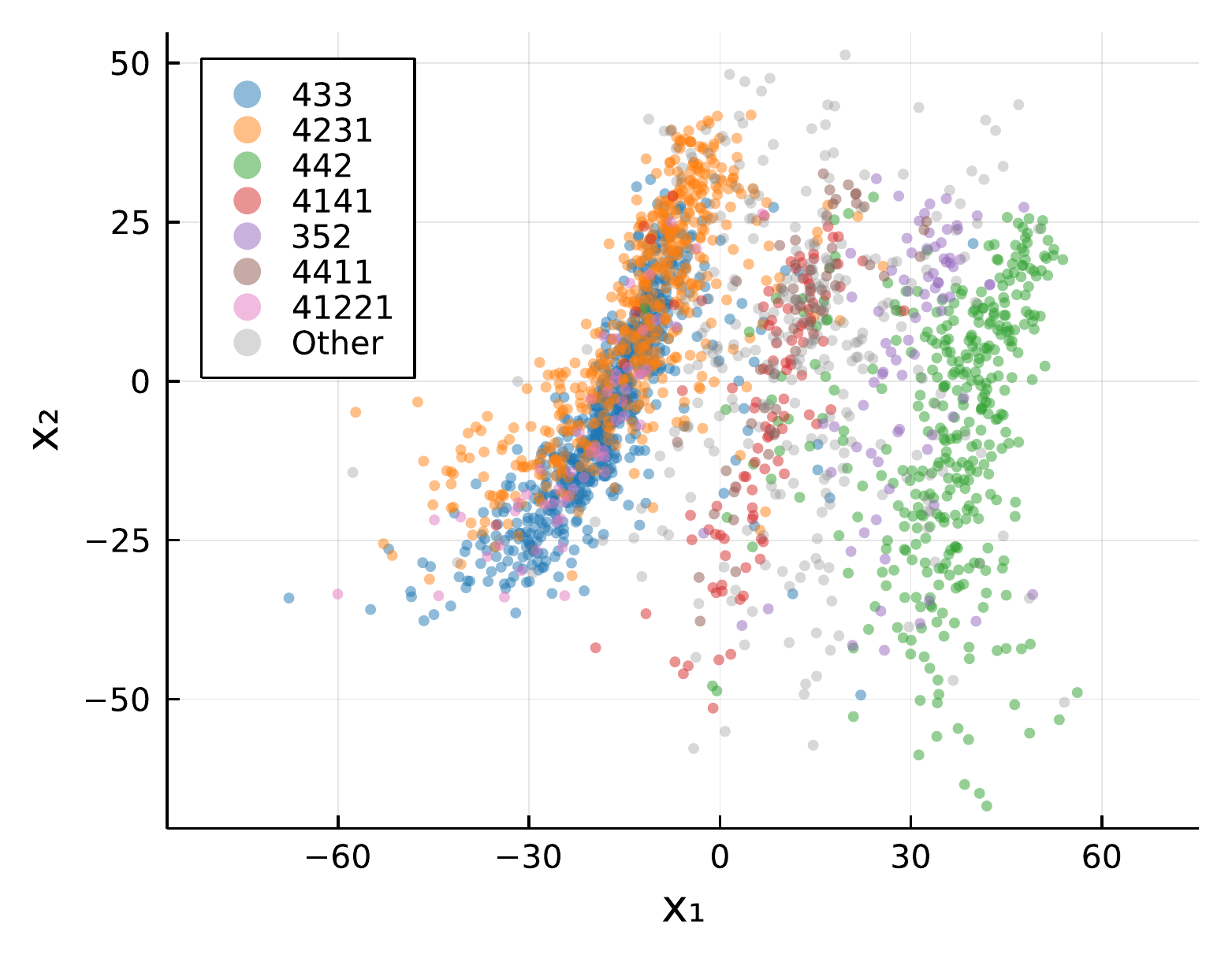}
	    \caption{$d_{\mathrm{DTW}}$}
	    \label{fig:football_embedding_dtw}
	\end{subfigure}
	\begin{subfigure}{0.48\linewidth}
	    \centering
	    \includegraphics[width=0.9\linewidth]{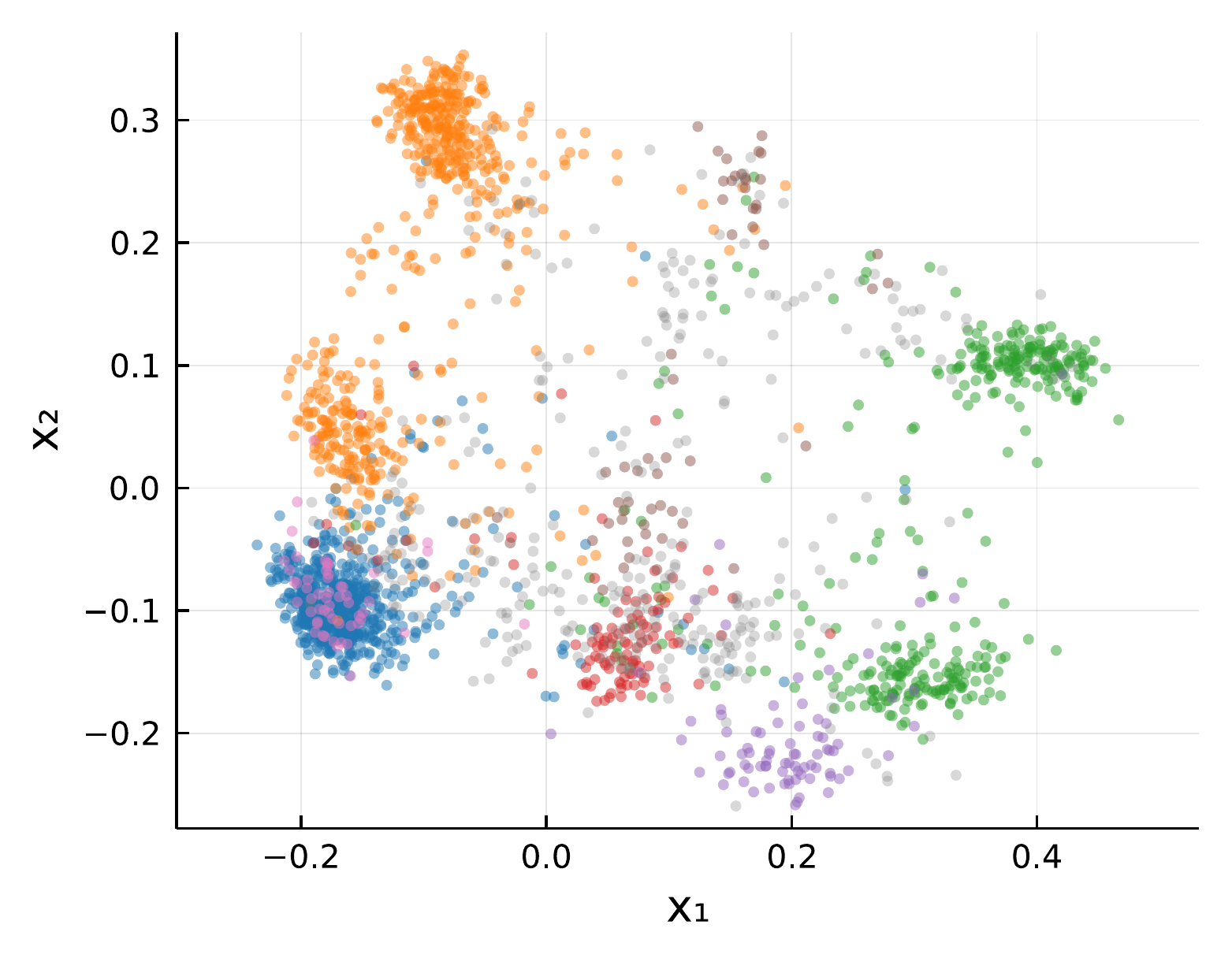}
	    \caption{$\bar{d}_{\mathrm{M},0.5}$}
	    \label{fig:football_embedding_match}
	\end{subfigure}
	\begin{subfigure}{0.48\linewidth}
	    \centering
	    \includegraphics[width=0.9\linewidth]{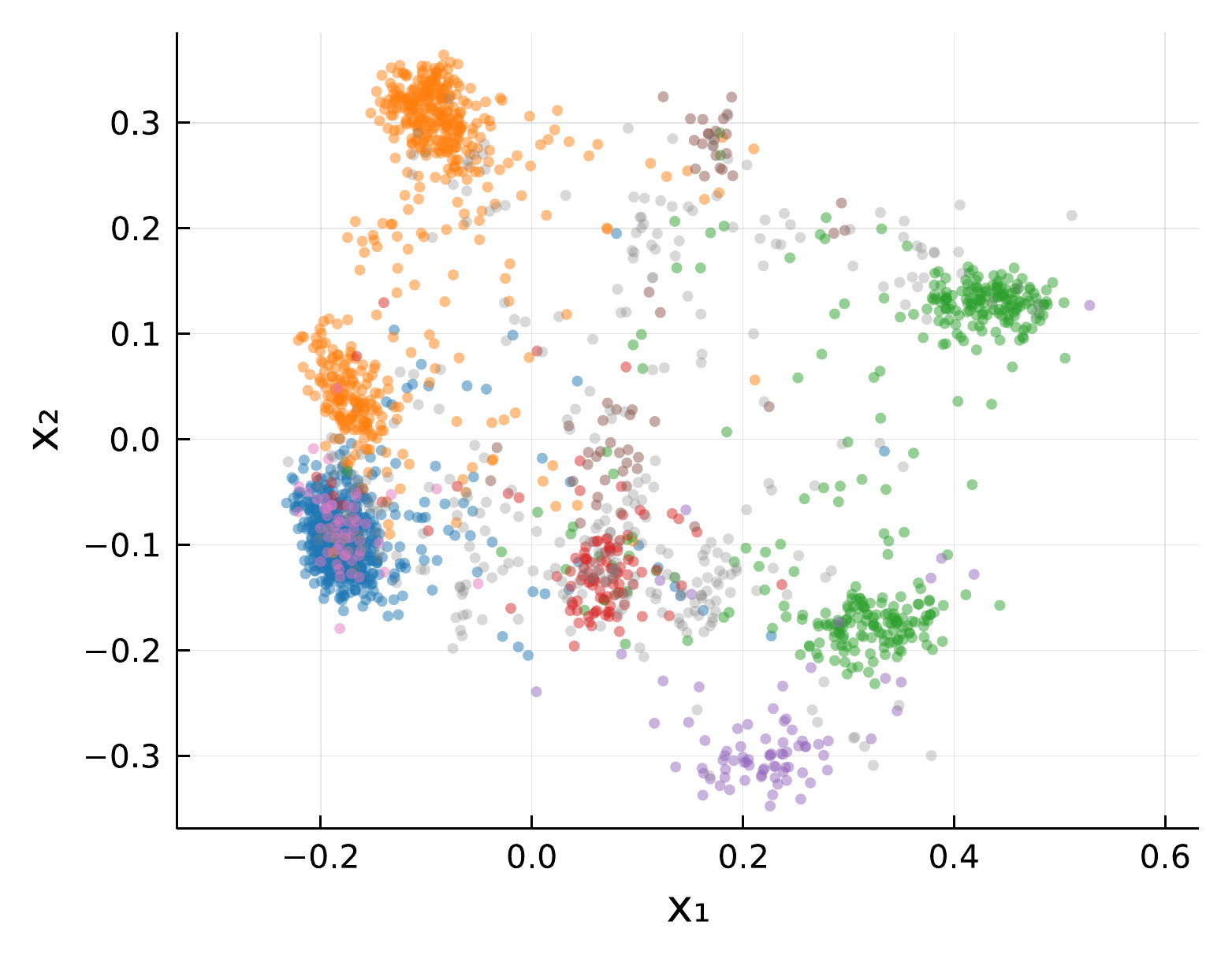}
	   \caption{$d_{\mathrm{EMD}}$}
	   \label{fig:football_embedding_emd}
	\end{subfigure}
	\caption{Embeddings of football matches via different distance measures. In each subplot, data points corresponds to a particular team in a given match, labelled to indicate the formation in which a team spent the most time, where for example "433" denotes a formation consisting of 4 defenders, 3 midfielders and 3 attackers. Here (a) and (b) show embeddings obtained via sequence distances, namely the the (normalised) fixed-penalty edit and DTW distances, whilst (c) and (d) show those obtained via two multiset distances, in particular, the (normalised) fixed-penalty matching and EMD distances, respectively.}
	\label{fig:football_embeddings}
\end{figure}

% \begin{figure}
%     \centering
%     \includegraphics[width=0.99\linewidth]{img/data_analysis/football_embed_all_for_paper.pdf}
%     \caption{Comparing embeddings of football matches via different distance measures. In each subplot, each data point corresponds to a particular team in a given match. In the top row, we have emebdding for two multiset distances, namely $\bar{d}_{\mathrm{M},0.5}$ and $d_{\mathrm{EMD}}$ respecively, whilst the bottom row shows embeddings obtained via sequence distance, showing $\bar{d}_{\mathrm{E},0.5}$ and $d_{\mathrm{DTW}}$ respectively. }
    
%     \label{fig:football_ebed_init_labs}
% \end{figure}

\section{Discussion}

In this paper, we have considered the problem of measuring the dissimilarity of sequences and multisets. Drawing on the wider literature, we have discussed various distances one can invoke, all of which make use of a pre-specified ground distance over the underlying space. For each distance, we have given a high-level intuition, proved theoretical properties and outlined how they can be computed. For certain distances, such as the EMD and DTW distance, we also propose extensions which allow the distances to satisfy additional metric conditions. Finally, we have illustrated how these distances can be used in practice through a novel analysis of an in-play football data set shared by StatsBomb, where we are able to uncover the squad formation based purely on passes between players.

Regarding future work, one could firstly consider whether other distances could be (or have been) defined. For example, can we consider an analogue of the EMD distance for sequences? One could also study through simulation what features each distance can take into account, providing guidance on which distance to use for a given problem or question of interest. There is also scope to expand on the data analysis of this work. For example, one could consider measuring quantitatively the relationship between formation and distance in the StatsBomb data, similar to the analysis of \cite{Donnat2018b}. Finally, note that often one can first aggregate observations to some other form before measuring their distance. As such, it is natural to ask whether one gains anything by using the distances discussed here? For example, the football data could be collapsed to a vector of counts over player positions, or perhaps a multigraph, encoding the number of passes observed between each pair of positions. Distances between these aggregates are likely to be much faster to compute than those between a sequence or multiset of paths, however, there is going to be a loss of information incurred through aggregation. It would be interesting to explore, either through a simulation study or real-data analysis, whether one gains something by taking into account this extra information via use of a multiset, or sequence, distance instead of an aggregate-based distance.

%% if your bibliography is in bibtex format, uncomment commands:

\bibliographystyle{imsart-nameyear} % Style BST file
\bibliography{library}
%\bibliography{bibliography}       % Bibliography file (usually '*.bib')

\newpage
\begin{appendix}

\begin{table}
    \centering
    \begin{tabular}{|l|l|}
    \hline
        \textbf{Position} & \textbf{Abbreviation} \\ \hline
        Center Back & CB \\ \hline
        Right Defensive Midfield & RDM \\ \hline
        Right Wing Back & RWB \\ \hline
        Right Center Back & RCB \\ \hline
        Goal Keeper & GK \\ \hline
        Left Center Back & LCB \\ \hline
        Left Wing Back & LWB \\ \hline
        Right Attacking Midfield & RAM \\ \hline
        Left Attacking Midfield & LAM \\ \hline
        Center Forward & CF \\ \hline
        Left Defensive Midfield & LDM \\ \hline
        Right Back & RB \\ \hline
        Right Wing & RW \\ \hline
        Left Wing & LW \\ \hline
        Center Defensive Midfield & CDM \\ \hline
        Right Center Midfield & RCM \\ \hline
        Left Back & LB \\ \hline
        Left Center Midfield & LCM \\ \hline
        Center Attacking Midfield & CAM \\ \hline
    \end{tabular}
    \vspace{0.2cm}
    \caption{Abbreviations for player positions used in \Cref{fig:football_ex}.}
    \label{table:position_abbrv}
\end{table}

\section{Distance computation}

\subsection{Matching distances}

\label{sec:computing_matching_distances}

As mentioned in \Cref{sec:matching_distances}, we consider evaluating both matching distances via the Hungarian algorithm \citep{Kuhn1955}, a specialised algorithm proposed to solve the so-called assignment problem. Suppose that one has two sets $A=\{a_1,\dots,a_n\}$ and $B=\{b_1,\dots,b_n\}$, both of size $n$, then assignment problem considers pairing elements of set $A$ with those of set $B$ in an `optimal' way, where the objective is defined by assigning a cost to each possible pairing. Note the labelling of elements here is arbitrary but will serve a purpose in what follows, allowing us to index set elements. 

A pairing of set elements can be encoded via a permutation $\sigma \in S_n$, where $S_n$ denotes the set of all permutation on $n$ symbols, with $\sigma(i)=j$ implying that $a_i \in A$ has been paired with $b_j \in B$. We summarise the cost of pairings in the $n \times n$ matrix $C$, where $C_{ij}$ denotes the cost incurred when $a_i \in A$ is paired with $b_j \in B$. Now, the assignment problem can be stated formally as finding the minimum cost permutation $\sigma$, that is
\begin{equation*}
    \begin{aligned}
    \min_{\sigma \in S_n} \sum_{i=1}^n C_{i, \sigma(i)}
    \end{aligned}
\end{equation*}
the solution of which may not be unique. Observe that though $A$ and $B$ are often assumed to be sets, this formulation works equally well if they are mulitsets, as we will assume them to be.

Given any square $n \times n$ matrix $C$, the Hungarian algorithm will return a permutation $\sigma$ which minimises the cost above. Towards evaluating $d_{\mathrm{M}}$ and $d_{\mathrm{M},\rho}$, we consider constructing a matrix $C$ such that the optimal solution found via the Hungarian algorithm coincides with that required in their respective definitions. Note that due to \Cref{thm:matching_complete,thm:fp_matching_complete}, there are situations in which we need only optimise over complete matchings. In these situations, we can minimise the size of optimisation problem to be solved via the Hungarian algorithm, or equivalently, minimise the size of $C$, as we outline in \Cref{sec:comp_matching_complete}. Alternatively, if one would like to optimise over all matchings, a slightly bigger $C$ can be specified, as detailed in \Cref{sec:comp_matching_all}. 

Given two multisets $\setx$ and $\sety$ the choice of which approach to take for each matching distance can be summarised as follows
\begin{itemize}
    \item For $d_{\mathrm{M}}(\setx,\sety)$, optimise over complete matchings (\Cref{sec:comp_matching_complete})
    \item For $d_{\mathrm{M},\rho}(\setx,\sety)$, the approach depends on whether $\rho$ satisfies the condition of \Cref{thm:fp_matching_complete}. In particular, letting $K = \max_{\elmtx \in \setx, \elmty \in \sety}d(x,y)$ we have
    \begin{itemize}
        \item If $\rho \geq K/2$, optimise over complete matchings (\Cref{sec:comp_matching_complete})
        \item If $\rho < K/2$, optimise over all matchings (\Cref{sec:comp_matching_all}).
    \end{itemize}
\end{itemize}

% The approach we take here depends on whether we want to optimise over (i) all matchings or (ii) complete matchings; where since (ii) can be formulated as a simpler assignment problem it is the favourable choice. Thanks \Cref{thm:matching_complete,thm:fp_matching_complete}, we know in which cases optimising for over complete matchings will return the correct result.

\subsubsection{Optimising over complete matchings}

\label{sec:comp_matching_complete}
For multisets $\setx$ and $\sety$, suppose without loss of generality we have $|\setx| \leq |\sety|$. We construct the $|\sety| \times |\sety|$ matrix $C$ as follows
\begin{equation}
    \begin{aligned}
    C_{ij} = \begin{cases}
    d(x_i,y_j) & \text{if } i \leq |\setx| \\
    \lambda(y_i) & \text{if } i > |\setx|
    \end{cases}
    \end{aligned}
    \label{eq:cost_matrix_matching_general}
\end{equation}
where $\lambda(\elmty)$ denotes the penalty incurred when $\elmty \in \sety$ is not included in the matching, where for the matching distance $d_{\mathrm{M}}$ we let $\lambda(y) = d(y,\Lambda)$, whilst for the fixed-penalty matching distance $d_{\mathrm{M},\rho}$ we let $\lambda(y) = \rho$. Here, to account for the fact that $\setx$ and $\sety$ may be of different sizes, we effectively
introduce $|\sety| - |\setx|$ dummy elements which entries from $\sety$ can be paired with, where being paired with a dummy element is equivalent to being un-matched. Now, each $\sigma \in S_{|\sety|}$ encodes a matching $\M$ of $\setx$ and $\sety$ given by the following
$$\M = \{(x_i,y_{\sigma(i)}) \, : \, 1 \leq i \leq |\setx|\},$$
which includes all elements of $\setx$ and is thus complete. Moreover, according to $C$ this has the following cost
\begin{equation*}
    \begin{aligned}
    C(\M) &= \sum_{i=1}^{|\sety|} C_{i, \sigma(i)}\\
    &= \sum_{(x,y) \in \M} d(x,y) + \sum_{y \in \M_\sety^c} \lambda(y)\\
    &= \sum_{(x,y) \in \M} d(x,y) + \sum_{x \in \M_\setx^c} \lambda(x) + \sum_{y \in \M_\sety^c} \lambda(y) 
    \end{aligned}
\end{equation*}
where the last line follows since $\M_\setx^c=\emptyset$ by virtue of $\M$ being complete. As such, in optimising over the permutations $\sigma$ via the Hungarian algorithm we are effectively enacting the following optimisation
\begin{equation*}
    \begin{aligned}
    \min_{\M} \left\{\sum_{(x,y) \in \M} d(x,y) + \sum_{x \in \M_\setx^c} \lambda(x) + \sum_{y \in \M_\sety^c} \lambda(y) \right\}
    \end{aligned}
\end{equation*}
where $\M$ is a \textit{complete} matching. Observe that, with the respective $\lambda(\cdot)$ substituted in, this is almost identical to optimisation of \Cref{def:matching,def:fp_matching}, the only difference being that this optimises over only complete matchings. 

\subsubsection{Optimising over all matchings}

\label{sec:comp_matching_all}

For multisets $\setx$ and $\sety$, in this case we construct the $(|\setx|+|\sety|) \times (|\setx|+|\sety|)$ matrix $C$ as follows
\begin{equation*}
    \begin{aligned}
    C_{ij} = \begin{cases}
    d(x_i,y_j) & \text{if } i \leq |\setx| \text{ and } j \leq |\sety| \\
    \lambda(y_i) & \text{if } i > |\setx| \text{ and } j \leq |\sety|\\
    \lambda(x_i) & \text{if } i \leq |\setx| \text{ and } j > |\sety| \\
    0 & \text{if } i > |\setx| \text{ and } j > |\sety|
    \end{cases}
    \end{aligned}
    \label{eq:cost_matrix_matching_general}
\end{equation*}
where again $\lambda(\cdot)$ denotes the penalisation for un-matched entries as stated in \Cref{sec:comp_matching_complete}. In this case, we introduce dummy elements to \textit{both} sets, namely $|\sety|$ in $\setx$ and $|\setx|$ in $\sety$, implying now elements from either set can be un-matched by being paired with a dummy element. Each $\sigma \in S_{|\setx|+|\sety|}$ encodes a matching $\M$ of $\setx$ and $\sety$ given by the following
$$\M = \{(x_i,y_{\sigma(i)}) \, : \, 1 \leq i \leq |\setx|, \, \sigma(i) \leq |\sety| \},$$
where we must include the extra constraint $\sigma(i) \leq |\sety|$ since in this case elements of $\setx$ can be paired with dummy elements, that is, be un-matched. By the definition of $C$, this implies the following cost
\begin{equation*}
    \begin{aligned}
    C(\M) &= \sum_{i = 1}^{|\setx| + |\sety|} C_{i,\sigma(i)} \\
    &= \sum_{(x,y) \in \M} d(x,y) + \sum_{x \in \M_\setx^c} \lambda(x) + \sum_{y \in \M_\sety^c} \lambda(y) 
    \end{aligned}
\end{equation*}
where $\M$ is a matching (not necessarily complete). Again, a comparison with \Cref{def:matching,def:fp_matching} reveals the similarity between the minimisation problems therein and that of the assignment problem parameterised by $C$. Moreover, in this case since we are optimising over \textit{all} matching they are indeed equivalent.
\subsection{Edit distances}

\label{sec:computing_edit_distance}

Both edit distances are special cases of the so-called string edit distance of \cite{Wagner1974}, and as such the dynamic programming algorithm proposed therein can be invoked to compute them. Consider first the edit distance (\Cref{def:edit}). Supposing that $\seqx$ and $\seqy$ are the sequences to be compared, introducing the notation $\seqx_{k:l} = (\elmtx_k, \dots, \elmtx_l)$, the approach is to evaluate $d_{\mathrm{E}}(\seqx_{1:i}, \seqy_{1:j})$ incrementally until $i=|\seqx|$ and $j=|\seqy|$ using the following recursive result
\begin{equation*}
    \begin{aligned}
    d_{\mathrm{E}}(\seqx_{1:i},\seqy_{1:j}) = \min\begin{cases}
    d_{\mathrm{E}}(\seqx_{1:(i-1)}, \seqy_{1:j}) + d(x_i,\Lambda)\\
    d_{\mathrm{E}}(\seqx_{1:i}, \seqy_{1:(j-1)}) + d(y_j,\Lambda) \\ 
    d_{\mathrm{E}}(\seqx_{1:(i-1)}, \seqy_{1:(j-1)}) + d(x_i,y_j)
    \end{cases}
    \end{aligned}
\end{equation*}
where here we are essentially comparing three possibilities (i) the $i$th entry of $\seqx$ is un-matched, (ii) the $j$th entry of $\seqy$ is un-matched, and (iii) the $i$th entry of $\seqx$ is matched with the $j$th entry of $\seqy$. Introducing the notation $C_{ij} = d_{\mathrm{E}}(\seqx_{1:(i-1)},\seqy_{1:(j-1)})$ this is equivalent to filling up the matrix $C$ either row-by-row or column-by-column via the following recursive formula 
\begin{equation*}
    \begin{aligned}
    C_{(i+1)(j+1)} = \min\begin{cases}
    C_{i(j+1)} + d(x_i,\Lambda)\\
    C_{(i+1)j} + d(y_j,\Lambda) \\ 
    C_{ij} + d(x_i,y_j)
    \end{cases}
    \end{aligned}
\end{equation*}
where the final entry corresponds to the desired distance, that is $d_{\mathrm{E}}(\seqx,\seqy)=C_{(|\seqx|+1)(|\seqy|+1)}$.

Note we add one to all indices here since the first column and row of $C$ function as boundary conditions. These correspond to when $i=1$ or $j=1$, that is, when we have values such as $\seqx_{1:0}$ or $\seqy_{1:0}$ appearing in the recursive definition. Towards specifying these values, we see $\seqx_{1:0}$ as an empty sequence, so that when comparing $\seqx_{1:0}$ to $\seqy_{1:j}$ each entry of the latter will be un-matched and hence penalised. This implies
$$C_{1(j+1)}=d_{\mathrm{E}}(\seqx_{1:0},\seqy_{1:j}) = \sum_{k=1}^j d(y_k,\Lambda),$$
for $j=1,\dots,|\seqy|$, whilst by equivalent reasoning we have
$$C_{(i+1)1}=d_{\mathrm{E}}(\seqx_{1:i},\seqy_{1:0}) = \sum_{k=1}^i d(x_k,\Lambda),$$
for $i=1,\dots,|\seqx|$. Finally, we let 
$$C_{11}= d_{\mathrm{E}}(\seqx_{1:0},\seqy_{1:0}) = 0$$
since $\seqx_{1:0}=\seqy_{1:0}$ by virtue of both being the empty sequence.

\Cref{alg:edit_distance} outlines pseudocode for the resulting algorithm to evluate $d_\mathrm{E}$, using this matrix notation. Furthermore, observe that when updating a row (or column) of $C$ one only needs to know the previous row (or column). As such, one need only store the current and previous row, leading to an algorithm which uses less memory and is typically faster. Pseudocode of this light-memory alternative can also be seen in \Cref{alg:edit_distance_light}.

Turning now to the fixed-penalty edit distance (\Cref{def:fp_edit}), the approach more-or-less the same, up to a slight change of the recursive formula. In particular, in this case we have
\begin{equation*}
    \begin{aligned}
    d_{\mathrm{E},\rho}(\seqx_{1:i},\seqy_{1:j}) = \min\begin{cases}
    d_{\mathrm{E},\rho}(\seqx_{1:(i-1)}, \seqy_{1:j}) + \rho \\ 
    d_{\mathrm{E},\rho}(\seqx_{1:i}, \seqy_{1:(j-1)}) + \rho\\
    d_{\mathrm{E},\rho}(\seqx_{1:(i-1)}, \seqy_{1:(j-1)}) + d(x_i,y_j)
    \end{cases}
    \end{aligned}
\end{equation*}
which leads to an analogous definition of matrix $C$, with its corresponding recursive formula given by  
\begin{equation*}
    \begin{aligned}
    C_{(i+1)(j+1)} = \min\begin{cases}
    C_{i(j+1)} +\rho\\
    C_{(i+1)j} + \rho \\ 
    C_{ij} + d(x_i,y_j)
    \end{cases}
    \end{aligned}
\end{equation*}
with $C_{(|\seqx|+1)(|\seqy|+1)}$ again corresponding to the desired distance. Moreover, in this case we have 
\begin{align*}
    C_{(i+1)1}&=i\rho \quad(\text{for }i=0,\dots,|\seqx|)\\
     C_{1(j+1)}&=j\rho \quad(\text{for }j=0,\dots,|\seqy|).
\end{align*}
Pseudocode of the resulting of the resulting algorithm to evaluate $d_{\mathrm{E},\rho}(\seqx,\seqy)$ can be seen in \Cref{alg:fp_edit_distance}, with the light-memory analogue outlined in \Cref{alg:fp_edit_distance_light}.

\subsection{Dynamic time warping distances}

\label{sec:computing_dtw_distance}

Similar to the edit distances, the DTW distances can be evaluated via dynamic programming. In fact, the algorithms are almost identical, differing only in the recursive formulae used. 

First, we outline how to compute $d_{\mathrm{DTW}}(\seqx,\seqy)$ for given sequences $\seqx$ and $\seqy$, following the implementation of \cite{Gold2018}, Sec. 3. Using the notation $\seqx_{k:l} = (\elmtx_k, \dots, \elmtx_l)$, one evaluates $d_{\mathrm{DTW}}(\seqx_{1:i}, \seqy_{1:j})$ incrementally until $i=|\seqx|$ and $j=|\seqy|$ via the following recursive result
\begin{equation*}
    \begin{aligned}
    d_{\mathrm{DTW}}(\seqx_{1:i},\seqy_{1:j}) = d(\elmtx_i,\elmty_j) + \min\begin{cases}
    d_{\mathrm{DTW}}(\seqx_{1:(i-1)}, \seqy_{1:j})\\
    d_{\mathrm{DTW}}(\seqx_{1:i}, \seqy_{1:(j-1)}) \\
    d_{\mathrm{DTW}}(\seqx_{1:(i-1)}, \seqy_{1:(j-1)})
    \end{cases}
    \end{aligned}
\end{equation*}
where here one is essentially comparing three possibilities (i) warping on the $j$th entry of $\seqy$, that is, $\elmty_j$ being paired with more than one element of $\seqx$, (ii) warping on the $i$th entry of $\seqx$, and (iii) no warping, with $\elmtx_i$ and $\elmty_j$ being paired \textit{only} with each other. Note the $d(\elmtx_i,\elmty_j)$ term comes out front of the minimisation since by definition $\elmtx_i$ and $\elmty_j$ must be paired.

Introducing the notation $C_{ij}= d_{\mathrm{DTW}}(\seqx_{1:(i-1)},\seqy_{1:(j-1)})$, the incremental computation can be seen as filling-up the matrix $C$ either row-by-row or column-by-column via the following recursive formula
\begin{equation*}
    \begin{aligned}
    C_{(i+1)(j+1)} = d(\elmtx_i,\elmty_j) + \min \begin{cases}
    C_{i(j+1)} \\
    C_{(i+1)j} \\
    C_{ij}
    \end{cases}
    \end{aligned}
\end{equation*}
with $d_{\mathrm{DTW}}(\seqx,\seqy)= C_{(|\seqx|+1)(|\seqy|+1)}$. As with evaluating the edit distances (\Cref{sec:computing_edit_distance}), we must also pre-specify the first row and column on $C$. Here we again assume $C_{11}=0$ whilst 
\begin{align*}
    C_{(i+1)1} = \infty \quad (\text{for }i=1,\dots,|\seqx|) && C_{1(j+1)} = \infty \quad (\text{for }j=1,\dots,|\seqy|).
\end{align*}
To see why this is the case, consider the second column entries, that is$$C_{i2} = d_{\mathrm{DTW}}(\seqx_{1:(i-1)},\seqy_{1:1}).$$
Observe that since $\seqy_{1:1}=(\elmty_1)$ is a sequence with a single entry, the only valid coupling here is where $\elmty_1$ is paired with every entry of $\seqx_{1:(i-1)}$. By opting for this choice of boundary values for $C$ one essentially ensures this occurs via the recursive formula. In particular, if one considers filling the second column of $C$, one has
\begin{equation*}
    \begin{aligned}
    C_{22} = d(\elmtx_1,\elmty_1) + \min \begin{cases}
    \infty \\
    \infty \\
    0
    \end{cases}
    \end{aligned}
\end{equation*}
so we choose the third option, paring the first two entries with no warping, whilst for $i>2$ we have
\begin{equation*}
    \begin{aligned}
    C_{i2} = d(\elmtx_1,\elmty_1) + \min \begin{cases}
    C_{(i-1)2} \\
    \infty \\
    \infty
    \end{cases}
    \end{aligned}
\end{equation*}
where here we choose the first option, which corresponds to warping on the 1st entry of $\seqx$, that is, $\elmty_1$ being paired with more than one element of $\seqx$. The same reasoning can be used to justify the intial values of the first row by considering filling the second row of $C$, wherein essentially the roles of $\seqx$ and $\seqy$ are swapped. 

\Cref{alg:dtw_distance} outlines the algorithm which fills the matrix $C$ via this recursive formula to obtain the desired distance. As with the edit distances, this procedure only requires knowledge of the previous and current row, and hence a lighter memory alternative can be considered, as detailed in \Cref{alg:dtw_distance_light}.

Turning now to the fixed-penalty DTW distance $d_{\mathrm{DTW},\rho}$, the approach is almost identical, albeit with a slight change in the recursive formula. Namely, to evaluate $d_{\mathrm{DTW},\rho}(\seqx_{1:i},\seqy_{1:j})$ we use the following 
\begin{equation*}
    \begin{aligned}
    d_{\mathrm{DTW},\rho}(\seqx_{1:i},\seqy_{1:j}) = d(\elmtx_i,\elmty_j) + \min\begin{cases}
    d_{\mathrm{DTW},\rho}(\seqx_{1:(i-1)}, \seqy_{1:j}) + \rho\\
    d_{\mathrm{DTW},\rho}(\seqx_{1:i}, \seqy_{1:(j-1)})+ \rho \\
    d_{\mathrm{DTW},\rho}(\seqx_{1:(i-1)}, \seqy_{1:(j-1)})
    \end{cases}
    \end{aligned}
\end{equation*}
where we have simply included the $\rho$ term in the cases corresponding to warping. This leads to analogous algorithms to compute $d_{\mathrm{DTW},\rho}(\seqx,\seqy)$, namely \Cref{alg:fp_dtw_distance}, which does so by incrementally filling a matrix, and the light-memory approach of \Cref{alg:fp_dtw_distance_light}, which stores only the current and previous row.

% \subsection{Towards a General Coupling Distance for Multisets}

% The coupling distance defined is one-way, in that it considers only the interactions of the smaller set being paired to more than one of the larger set. One might consider whether this could be generalised to be two-way. 

% Denote decision variables $z_{ij} \in \{0,1\}$, where $z_{ij} = 1$ if $\elmtx_i$ is paired with $\elmty_j$, then with $d_{ij}=d(\elmtx_i, \elmty_j)$ we are essentially looking to solve the following integer program
% \begin{equation}
%     \begin{aligned}
%         \min &\sum_{i,j} d_{ij} z_{ij} \\
%         \text{s.t.} \, &\sum_{i=1}^N z_{ij} \geq 1 \\
%         &\sum_{j=1}^M z_{ij} \geq 1
%     \end{aligned}
% \end{equation}

\section{Proofs}

\label{sec:metric_proofs}

\subsection{Multiset distances}

\begin{proof}[Proof of \Cref{thm:matching_is_metric} (Part 1)]
    \label{proof:matching_is_metric}
	To aide this exposition, we write $d_{\mathrm{M}}$ in terms of its associate cost function as follows
	\begin{equation*}
		\begin{aligned}
			d_{\mathrm{M}}(\setx,\sety) = \min_{\M} C(\M)
			\label{dist:mathcing}
		\end{aligned}
	\end{equation*}
	where
	\begin{equation*}
	    \begin{aligned}
	        C(\M) =  \left(  \sum_{(\elmtx,\elmty) \in \M} d(\elmtx,\elmty) \right) + \sum_{\elmtx \in \M_\setx^c} d(\elmtx, \Lambda) + \sum_{\elmty \in \M_\sety^c} d(\elmty, \Lambda),
	    \end{aligned}
	\end{equation*}
	denoting the cost of the matching $\M$. 

	We first show \eqref{metric_condition2:indentity} holds. Assuming $X = Y$, then one can construct a matching $\M_1$ by pairing equivalent elements of $X$ and $Y$, leading to the following upper bound on the distance 
	\begin{equation}
		\begin{aligned}
			d_{\mathrm{M}}(X,Y) &\leq C(\M_1) \\
			&= \sum_{(x,y) \in \M_1} d(x,y) + 0 + 0 \\
			&= 0
		\end{aligned}
	\end{equation}
	where the second line follow since $\M_1$ includes all elements of $X$ and $Y$ and thus no penalisation will occur, whilst the final line follows since $\M_1$ matches equivalent elements and hence (using the fact $d(\cdot,\cdot)$ is a metric) all pairwise distances will be zero. Note also, since $d_{\mathrm{M}}$ will be a sum of positive values (since $d(\cdot,\cdot)$ is a metric), we also have $d_{\mathrm{M}}\geq 0$. Together this implies $d_{\mathrm{M}}(X,Y)=0$. 
	
	Conversely, assume that $d_{\mathrm{M}}(X,Y)=0$. This implies that both the matching cost and penalisation terms must be zero. Since we implicitly assume no element is equal to the null element $\Lambda$, then the penalty term being zero implies that all elements of $X$ and $Y$ must be included in the matching. Therefore, we have a complete matching with zero cost. Specifically, supposing $\M^*$ is the optimal matching, we have 
	\begin{equation*}
	    \begin{aligned}
	        d_{\mathrm{M}}(\setx,\sety) &=\sum_{(x,y) \in \M^*}d(x,y) \\
	        &= 0.
	    \end{aligned}
	\end{equation*}
	Since $d$ is a metric it is non-negative, and hence
	$$ d(x,y) = 0, \quad \forall \, (x,y) \in \M^*,$$
	which, again using the fact $d$ is a metric, implies
	$$x = y, \quad, \forall \, (x,y) \in \M^*$$
	and hence $X=Y$, confirming satisfaction of \eqref{metric_condition2:indentity}.
	
	The condition \eqref{metric_condition3:symmetry} follows trivially from the symmetry of $d(\cdot,\cdot)$ and the penalisation term. 
	
	The final condition to show is the triangle inequality \eqref{metric_condition4:tri}. Assuming that $\setx=\{\elmtx_1,\dots,\elmtx_n\}$, $\sety=\{\elmty_1,\dots,\elmty_m\}$ and $Z = \{\elmtz_1,\dots,\elmtz_k\}$ are three multisets, we seek to show that 
	$$d_\mathrm{M}(X,Y) \leq d_\mathrm{M}(X,Z) + d_\mathrm{M}(Z,Y).$$
    Now, let $\M^*_1$ and $\M^*_2$ denote optimal matchings for $d_{\mathrm{M}}(X,Z)$ and $d_{\mathrm{M}}(Z,Y)$ respectively, so that 
    \begin{align*}
        d_{\mathrm{M}}(X,Z)=C(\M^*_1) && d_{\mathrm{M}}(Z,Y)=C(\M^*_2)
    \end{align*}
    writing these as
	\begin{equation*}
		\begin{aligned}
			\M_1^* = \{(x_{i_1}, z_{j_1}), \dots, (x_{i_r}, z_{j_r})\} && && \M_2^* = \{(z_{l_1}, y_{k_1}), \dots, (z_{l_s}, y_{k_s})\}.
		\end{aligned}
	\end{equation*}
	Now, $\M^*_1$ and $\M_2^*$ induce a matching $\M_3$ of $X$ and $Y$ as follows
	\begin{align}
		\M_3 = \{(x_i, y_j) \, : \, (x_i, z_k) \in \M^*_1 \text{ and } (z_k, y_j) \in \M^*_2 \text{ for some } z_k \in Z\}
		\label{eq:induced_matching}
	\end{align}
	that is, we pair elements of $X$ and $Y$ if they were paired to the same elements of $Z$. Notice by definition we have 
	\begin{equation*}
	    \begin{aligned}
	        d_{\mathrm{M}}(\seqx, \seqy) \leq C(\M_3),
	    \end{aligned}
	\end{equation*}
	consequently the triangle inequality will follow if we can show the following holds
	\begin{equation}
		\begin{aligned}
			C(\M_3) &\leq d_{\mathrm{M}}(X,Z) + d_{\mathrm{M}}(Z,Y) 
			\label{eq:matching_tri_ineq}.
		\end{aligned}
	\end{equation}
	
	To prove \cref{eq:matching_tri_ineq} we consider every possible term on the LHS and show that this is less than or equal to some unique terms appearing on the RHS. The keys terms appearing on the LHS are (i) pairwise distances for matched elements (ii) penalisations of unmatched elements. 
	
	We first consider (i). By definition of $\M_3$, each pair $(x_i, y_j) \in \M_3$ is associated with some \textit{unique} $(x_i, z_k) \in \M_1^*$ and $(z_k, y_j)\in \M_2^*$, that is, there is some element $z_k \in Z$ which both $x_i$ and $y_j$ are matched to. Furthermore, since $d(\cdot, \cdot)$ is a distance metric it satisfies the triangle inequality, and so
	$$d(x_i, y_j) \leq d(x_i, z_k) + d(z_k, y_j),$$
	and thus each pairwise distance of matched elements on the LHS of \cref{eq:matching_tri_ineq} is less than or equal to some unique terms on the RHS. 
	
	For (ii) consider first the penalisation terms for elements of $X$ not included in the matching $\M_3$, that is $d(x,\Lambda)$ for $x \in(\M_3)^c_X$. We now seek to show that $d(x, \Lambda)$ is less than or equal to some (unique) terms appearing on the RHS of \cref{eq:matching_tri_ineq}. For $x$ to not be in $\M_3$ one of two things must have happened 
	\begin{enumerate}
		\item $(x,z) \in \M_1^*$ for some $z \in Z$ with $(z, y) \not\in \M_2^*$ for any $y \in Y$
		$$\implies \text{ a term on the RHS of } d(x,z) + d(z,\Lambda)$$
		which will also be unique to the pair $(x,z)$. Now, since $d(\cdot,\cdot)$ is a metric it obeys the triangle inequality, thus
		$$d(x,\Lambda) \leq d(x,z) + d(z, \Lambda)$$
		as desired;
		\item Alternatively, we might have $(x,z) \not \in \M_1^*$ for any $z \in Z$
		$$\implies  \text{ a term on the RHS of } d(x,\Lambda),$$
		and thus in this case we trivially have
		$$d(x, \Lambda) \leq d(x,\Lambda).$$
	\end{enumerate}
	In either case, we have a term on the LHS of \cref{eq:matching_tri_ineq} which is less than or equal to some unique terms on the RHS. This argument can be applied similarly to the penalisation terms for elements of $Y$ not in the matching $\M_3$. 
	
	Thus we have that every term on the LHS of \cref{eq:matching_tri_ineq} is less than or equal to some unique terms on the RHS, proving the inequality holds. Thus $d_{\mathrm{M}}$ satisfies condition \eqref{metric_condition4:tri}, completing the proof. 
\end{proof}

\begin{proof}[Proof of \Cref{thm:fp_matching_is_metric} (Part 2)]
    To aide this exposition, we write $d_{\mathrm{M},\rho}$ in terms of its associate cost function as follows
	\begin{equation*}
		\begin{aligned}
			d_{\mathrm{M},\rho}(\setx,\sety) = \min_{\M} C(\M)
			\label{dist:mathcing}
		\end{aligned}
	\end{equation*}
	where
	\begin{equation*}
	    \begin{aligned}
	        C(\M) =  \sum_{(\elmtx,\elmty) \in \M} d(\elmtx,\elmty) + \rho( |\setx| + |\sety| - 2|\M| ) ,
	    \end{aligned}
	\end{equation*}
	denoting the cost of the matching $\M$.

	We first consider condition \eqref{metric_condition2:indentity}. Note firstly that since $d(x, y) \geq 0$ (as it is a metric) and $\rho >0$, this implies $d_{\mathrm{M},\rho}(X,Y) \geq 0$ for all multisets $X$ and $Y$. Now, assuming $X=Y$, one can construct a matching $\M_1$ by pairing equivalent elements of $X$ and $Y$. The existence of this matching leads to the following upper bound on the distance 
	\begin{equation*}
		\begin{aligned}
			d_{\mathrm{M},\rho}(X,Y) &\leq C(\M_1) \\
			&= \sum_{(x,y) \in \M_1} d(x,y) \\
			&= 0
		\end{aligned}
	\end{equation*}
	where the second line follow since $\M_1$ included all elements of $X$ and $Y$ and thus no penalisation will occur, whilst the final line follows since $\M_1$ matches equivalent elements and hence (using the fact that $d(\cdot,\cdot)$ is a metric) all pairwise distances will be zero. This combined with $d_{\mathrm{M},\rho}(\setx,\sety)\geq 0$ implies $d_{\mathrm{M},\rho}(X,Y)=0$. 
	
	Conversely, $d_{\mathrm{M},\rho}(X,Y)=0$ implies both the sum of pairwise distances and penalisation terms must be zero. Thus, if $\M^*$ is the optimal matching, then all elements of $X$ and $Y$ are included in $\M^*$ and we have
	$$d_{\mathrm{M},\rho}=\sum_{(x,y) \in \M^*} d(x,y) = 0.$$
	Now, since $d(\cdot, \cdot)\geq 0$ this implies
	$$d(x,y) = 0, \quad \forall \, (x,y) \in \M^*,$$
	and hence
	$$x = y, \quad \forall \, (x,y) \in \M^*.$$
	Since all elements of either set are included in $\M^*$ this implies $X=Y$, thus confirming satisfaction of \eqref{metric_condition2:indentity}. 
	
	As with \Cref{thm:matching_is_metric} (Part 1), the symmetry condition \eqref{metric_condition3:symmetry} follows trivially from the symmetry of $d(x,y)$ and the penalisation term. 
	
	Finally, we verify condition \eqref{metric_condition4:tri}, the triangle inequality. Here we follow exactly the same steps seen in the proof of \Cref{thm:matching_is_metric} (Part 1). Assuming that $\setx=\{\elmtx_1,\dots,\elmtx_n\}$, $\sety=\{\elmty_1,\dots,\elmty_m\}$ and $Z = \{\elmtz_1,\dots,\elmtz_k\}$ are three multisets, we seek to show that 
	$$d_{\mathrm{M},\rho}(X,Y) \leq d_{\mathrm{M},\rho}(X,Z) + d_{\mathrm{M},\rho}(Z,Y).$$
    Now, let $\M^*_1$ and $\M^*_2$ denote optimal matchings for $d_{\mathrm{M}, \rho}(X,Z)$ and $d_{\mathrm{M},\rho}(Z,Y)$ respectively, that is 
    \begin{align*}
        d_{\mathrm{M},\rho}(X,Z)=C(\M^*_1) && d_{\mathrm{M},\rho}(Z,Y)=C(\M^*_2)
    \end{align*}
    then these again induce a matching $\M_3$ of $\seqx$ and $\seqy$ as in \cref{eq:induced_matching}. Moreover, following the reasoning therein, the triangle inequality will follow if we can show the following inequality holds
	\begin{equation}
		\begin{aligned}
			C(\M_3) &\leq d_{\mathrm{M},\rho}(X,Z) + d_{\mathrm{M},\rho}(Z,Y) 
			\label{eq:fp_matching_tri_ineq},
		\end{aligned}
	\end{equation}
    where the only difference here is in the form of cost function. 
    
	As in the proof of \Cref{thm:matching_is_metric} (Part 1), the route we take here is to show that every term on the LHS of \cref{eq:fp_matching_tri_ineq} is less than or equal to some unique terms appearing on the RHS. Again, we have terms on the LHS of two types (i) pairwise distances for matched elements (ii) penalisation of unmatched elements. 
	
	The argument for showing terms of the form (i) are less than or equal to some unique terms on the RHS is exactly the same as that in the proof of \Cref{thm:matching_is_metric} (Part 1), relying on the fact that $d(\cdot, \cdot)$ is a metric and so obeys the triangle inequality. For brevity, we do not repeat this here.
	
	The argument for terms of the form (ii) is slightly different and so we give full exposition. Considering first the penalisation terms for elements of $X$ not included in the matching $\M_3$, we will have a $\rho$ for each $x \in  (\M_3)^c_X$. We now seek to show that each $\rho$ is less than or equal to some unique terms appearing on the RHS of \cref{eq:fp_matching_tri_ineq}. For $x$ to not be in $\M_3$ one of two things must have happened 
	\begin{enumerate}
		\item $(x,z) \in \M_1^*$ for some $z \in Z$ with $(z, y) \not\in \M_2^*$ for any $y \in Y$
		$$\implies \text{ a term on the RHS of } d(x,z) + \rho$$
		which will also be unique to the pair $(x,z)$. Now, since $d(\cdot,\cdot)$ is a metric it is non-negative, thus
		$$\rho \leq d(x,y) + \rho$$
		as desired;
		\item Alternatively, we might have $(x,z) \not \in \M_1^*$ for any $z \in Z$
		$$\implies  \text{ a term on the RHS of } \rho,$$
		and thus in this case we trivially have
		$$\rho \leq \rho.$$
	\end{enumerate}
	In either case, we again have a term on the LHS of \cref{eq:fp_matching_tri_ineq} which is less than or equal to some unique terms on the RHS. Moreover, thus argument will apply similarly to the penalisation terms for $Y$ without loss of generality. Thus we have that every term on the LHS of \cref{eq:fp_matching_tri_ineq} is less than or equal to some unique terms on the RHS, proving the inequality of \cref{eq:fp_matching_tri_ineq} holds. Thus $d_{\mathrm{M},\rho}$ satisfies condition \eqref{metric_condition4:tri}, completing the proof. 
\end{proof}

\begin{proof}[Proof of \Cref{thm:matching_complete}]

Given two multisets $\setx$ and $\sety$, let
\begin{equation*}
    \begin{aligned}
        C(\M) = \sum_{(\elmtx,\elmty) \in \M} d(x,y) + \sum_{\elmtx \not\in \M_\setx}d(\elmtx, \Lambda) + \sum_{\elmty \not\in \M_\sety} d(\elmty,\Lambda)
    \end{aligned}
\end{equation*}
and, towards proving this result, we assume that any matching $\M'$ for which 
$$C(\M')=\min_\M C(\M)=d_{\text{M}}(X,Y),$$
is \textit{not} complete, seeking a contradiction. There may be more than one such matching, so without loss of generality, let $\M'$ denote any one of these optimal matchings. Since $\M'$ is not complete, there must be a currently un-matched pair, that is, $(\elmtx^*,\elmty^*)$ such that $\elmtx^*\in \setx$ and $\elmty^* \in \sety$ but $\elmtx^* \not\in \M'_\setx$ and $\elmty^* \not\in \M'_\sety$. One can now define a new matching $\M''$ by augmenting $\M'$ as follows 
$$\M'' = \M' \cup \{(\elmtx^*, \elmty^*)\} $$
for which 
\begin{equation}
    \begin{aligned}
        C(\M'') &= \sum_{(\elmtx,\elmty) \in \M''} d(x,y) + \sum_{\elmtx \not\in \M''_\setx}d(\elmtx, \Lambda) + \sum_{\elmty \not\in \M''_\sety} d(\elmty,\Lambda) \\
        &= \sum_{(\elmtx,\elmty) \in \M'} d(x,y) + d(x^*,y^*) + \sum_{\elmtx \not\in \M''_\setx}d(\elmtx, \Lambda) + \sum_{\elmty \not\in \M''_\sety} d(\elmty,\Lambda) \\
        &\leq \sum_{(\elmtx,\elmty) \in \M'} d(x,y) + d(x^*,\Lambda) + d(y^*,\Lambda) + \sum_{\elmtx \not\in \M''_\setx}d(\elmtx, \Lambda) + \sum_{\elmty \not\in \M''_\sety} d(\elmty,\Lambda) \\
        &= \sum_{(\elmtx,\elmty) \in \M'} d(x,y) + \sum_{\elmtx \not\in \M'_\setx}d(\elmtx, \Lambda) + \sum_{\elmty \not\in \M'_\sety} d(\elmty,\Lambda)\\
        &= C(\M')
    \end{aligned}
    \label{eq:matching_complete_ineq}
\end{equation}
where in the third line we use the fact $d(\cdot,\cdot)$ is a distance metric, and hence obeys the triangle inequality. Since $\M'$ was optimal, we must also have $C(\M') \leq C(\M)$ for all matchings $\M$, which combined with \cref{eq:matching_complete_ineq} implies $C(\M'')=C(\M')$, that is, $\M''$ is also an optimal matching. Moreover, we have $|\M''| = |\M'| + 1$. Now, either (i) $\M''$ is complete, or (ii) we can repeat this augmentation, increasing the matching cardinality until it is complete. Either way, we arrive at a matching with is both optimal and complete, contradicting our assumption that all optimal matchings were not complete. The result now follows by contradiction.

\end{proof}
\begin{proof}[Proof of \Cref{thm:fp_matching_complete}]

Given two multisets $\setx$ and $\sety$, let
\begin{equation*}
    \begin{aligned}
        C(\M) = \sum_{(\elmtx,\elmty) \in \M} d(x,y) + \rho(n + m - 2|\M|)
    \end{aligned}
\end{equation*}
where $|\setx|=n$ and $|\sety|=m$, and letting 
$$K =  \max_{\elmtx\in \setx, \elmty\in\sety} d(\elmtx,\elmty)$$
we assume $\rho \geq K/2$. Towards proving the result, further assume any matching $\M'$ for which 
$$C(\M') = \min_\M C(\M) = d_{\mathrm{M},\rho}(X,Y)$$
is \textit{not} complete, seeking a contradiction. As in the proof of \Cref{thm:matching_complete}, without loss of generality we let $\M'$ denote one of these optimal matchings and define a new matching $\M''$ by augmenting $\M'$ with a presently un-matched pair $(x^*,y^*)$, that is 
$$\M'' = \M' \cup \{(x^*,y^*)\}$$
for which 
\begin{equation}
    \begin{aligned}
        C(\M'') &= \sum_{(\elmtx,\elmty) \in \M''} d(x,y) + \rho(n + m - 2|\M''|) \\
        &= \sum_{(\elmtx,\elmty) \in \M'} d(x,y) + d(x^*,y^*) - 2\rho +\rho(n + m - 2|\M'|)\\
        &\leq \sum_{(\elmtx,\elmty) \in \M'} d(x,y) + 2\rho - 2\rho + \rho(n + m - 2|\M'|)\\
        &=\sum_{(\elmtx,\elmty) \in \M'} d(x,y) + \rho(n + m - 2|\M'|) \\
        &= C(\M')
    \end{aligned}
    \label{eq:fp_matching_complete_ineq}
\end{equation}
where in the second line we use the fact that $|\M''|=|\M'|+1$, whilst in the third line we used the fact that 
$$d(x,y) \leq K \leq 2\rho $$
for any $\elmtx\in \setx$ and $\elmty \in \sety$, by definition of $K$ and the assumption regarding $\rho$. As in the proof of \Cref{thm:matching_complete}, since $\M'$ was assumed optimal, \cref{eq:fp_matching_complete_ineq} implies that $\M''$ must also be optimal. Moreover, either (i) $\M''$ is complete, or (ii) we may repeat this augmentation until it is. In either case, we arrive at a matching which is optimal and complete. Hence the result follows by contradiction.
\end{proof}

\begin{proof}[Proof of \Cref{thm:emd_metric_cond}]
    \label{proof:emd_metric_cond}
    In what follows we will use the notation $d_{\mathrm{W}_1}(\mu_\setx, \mu_\sety)$ for the 1-Wasserstein distances between the \textit{distributions} $\mu_\setx$ and $\mu_\sety$, which is known to be a distance metric  \cite[Prop. 2.2]{Peyre2019}. Observe that by our definition of the EMD between multisets (\Cref{def:emd}) we have $d_{\mathrm{EMD}}(\setx,\sety) = d_{\mathrm{W}_1}(\mu_\setx,\mu_\sety)$.
    
    The conditions \eqref{metric_condition3:symmetry} and \eqref{metric_condition4:tri} are inherited naturally. Firstly, we have
    \begin{equation*}
        \begin{aligned}
            d_{\mathrm{EMD}}(\setx,\sety) &= d_{\mathrm{W}_1}(\mu_\setx,\mu_\sety) \\
            &= d_{\mathrm{W}_1}(\mu_\sety,\mu_\setx) \\
            &= d_{\mathrm{EMD}}(\sety,\setx)
        \end{aligned}
    \end{equation*}
    where the second line follows since $d_{\mathrm{W}_1}$ is a metric between distributions, verifying that \eqref{metric_condition3:symmetry} holds. Secondly, for any multisets $\setx$, $\sety$ and $\setz$ we have 
    \begin{equation*}
        \begin{aligned}
            d_{\mathrm{EMD}}(\setx,\sety) &= d_{\mathrm{W}_1}(\mu_\setx,\mu_\sety)\\
            &\leq d_{\mathrm{W}_1}(\mu_\setx,\mu_\setz) + d_{\mathrm{W}_1}(\mu_\setz,\mu_\sety) \\
            &= d_{\mathrm{EMD}}(\setx,\setz) + d_{\mathrm{EMD}}(\setz,\sety)
        \end{aligned}
    \end{equation*}
    where again the second line follows since $d_{\mathrm{W}_1}$ is a metric. Thus \eqref{metric_condition4:tri} also holds. 
    
    We now assume metric condition \eqref{metric_condition2:indentity} holds, seeking a contradiction. To do so, let $\setx$ be a multiset and define $\sety$ via its multiplicity function as follows (for any $x \in \X$)
    $$m_\sety(x) = C \cdot m_\setx(x)$$
    where $C \in \mathbb{Z}_+$, that is, $\sety$ and $\setx$ are proportional. Observe that if $C>1$ then $\setx \not= \sety$ whilst 
    \begin{equation*}
        \begin{aligned}
            \mu_\sety(x) = \frac{m_\sety(x)}{|\sety|} = \frac{C \cdot m_\setx(x) }{C \cdot |\setx|} = \mu_\setx(x)
        \end{aligned}
    \end{equation*}
    for any $x \in \X$, that is, $\mu_\setx=\mu_\sety$. Consequently, we have $\setx \not= \sety$ and 
    $$d_{\mathrm{EMD}}(\setx,\sety)=d_{\mathrm{W}_1}(\mu_\setx,\mu_\sety)=0,$$ thus contradicting the assumption condition \eqref{metric_condition2:indentity} holds.
\end{proof}
\begin{proof}[Proof of \Cref{thm:ex_emd_is_metric}]
\label{proof:ex_emd_metric}
Firstly, since both $d_{\mathrm{EMD}}$ and $d_s$ satisfy metric conditions \eqref{metric_condition3:symmetry} and \eqref{metric_condition4:tri}, so will a linear combination thereof.

As in the proof of \Cref{thm:emd_metric_cond}, we use the notation $d_{\mathrm{W}_1}(\mu_\setx, \mu_\sety)$ for the 1-Wasserstein distance between the \textit{distributions} $\mu_\setx$ and $\mu_\sety$, known to be a distance metric  \cite[Prop. 2.2]{Peyre2019}. Furthermore, by our definition of the EMD between multisets (\Cref{def:emd}) we have $d_{\mathrm{EMD}}(\setx,\sety) = d_{\mathrm{W}_1}(\mu_\setx,\mu_\sety)$.

 Towards proving \eqref{metric_condition2:indentity} holds, assume that $\setx = \sety$, which implies $\mu_{\setx} = \mu_{\sety}$ and $|\setx|=|\sety|$. Since both $d_{\mathrm{W}_1}$ and $d_s$ are metrics this implies $d_{\mathrm{W}_1}(\mu_\setx, \mu_{\sety}) = d_s(|\setx|,|\sety|) = 0$, and so
$$d_{\mathrm{sEMD}}(\setx, \sety) =  \tau \cdot 0 + (1-\tau) \cdot 0 = 0.$$
Conversely, assume that $d_{\mathrm{sEMD}}(\setx, \sety) = 0$. Being a linear combination of non-neagtive terms, this implies both $d_{\mathrm{EMD}}(\setx, \sety)=0$, and $d_s(|\setx|,|\sety|) = 0$. Now, since $d_s$ is a metric we have $|\setx|=|\sety|$, whilst since $d_{\mathrm{EMD}}(\setx,\sety) = d_{\mathrm{W}_1}(\mu_\setx,\mu_\sety)$ we have $d_{\mathrm{W_1}}(\mu_\setx,\mu_\sety)=0$ which, since $d_{\mathrm{W}_1}$ is a metric, implies $\mu_\setx=\mu_\sety$, which together imply for any $\elmtx \in \X$ we must have 
\begin{equation*}
    \begin{aligned}
        \mu_\setx(\elmtx) &= \mu_\sety(\elmtx) \\
	\implies \frac{m_\setx(\elmtx)}{|\setx|} &= \frac{m_\sety(\elmtx)}{|\sety|} \\
	\implies m_\setx(\elmtx) &= m_\sety(\elmtx) \qquad (\text{since } |\setx|=|\sety|)
    \end{aligned}
\end{equation*}
and hence $m_\setx = m_\sety$, that is, $\setx = \sety$. This confirms \eqref{metric_condition2:indentity} and completes the proof. 
\end{proof}

\subsection{Sequence distances}

\label{sec:proofs_sequences}

\begin{remark}
	Both $d_{\mathrm{E}}$ and $d_{\mathrm{E},\rho}$ can be seen as special cases of the so-called \textit{string edit distance} proposed by \cite{Wagner1974}. We could, therefore, conclude right away that both are indeed distance metrics. However, for completeness, and to emphasise the close connections with the matching distances, we proceed to prove these results, emulating the structure and approach in the proof of \Cref{thm:matching_is_metric,thm:fp_matching_is_metric}.
\end{remark}

\begin{proof}[Proof of \Cref{thm:edit_distance_metric} (Part 1)]

    To aide this exposition, we write $d_{\mathrm{E}}(\seqx, \seqy)$ in terms of its cost function as follows
    \begin{align*}
		d_{\mathrm{E}}(\seqx,\seqy) = \min_{\M} C(\M)
	\end{align*}
	where $\M$ denotes a monotone matching of $\seqx$ and $\seqy$ and 
	\begin{equation*}
	    \begin{aligned}
	        C(\M) =  \sum_{(\elmtx,\elmty) \in \M} d(\elmtx,\elmty)  + \sum_{\elmtx \in \M_\setx^c} d(\elmtx, \Lambda) + \sum_{\elmty \in \M_\sety^c} d(\elmty, \Lambda)
	    \end{aligned}
	\end{equation*}
	denotes the cost of the matching $\M$.
	
	We first consider metric condition \eqref{metric_condition2:indentity} (identity of indiscernibles). Firstly, since $d$ is a metric we have $d(x,y) \geq 0$, which implies $d_\mathrm{E}(X,Y)\geq 0$ for any sequences $X$ and $Y$. Now, assuming that $X = Y$, then letting $n=|\seqx|=|\seqy|$ this implies 
	$$x_i = y_i \quad \text{ for } i=1, \dots, n.$$ 
	Consequently, we can trivially construct a monotone matching $\M^*$ which pairs equivalent entries
	\begin{align}
		\M^* = \{(x_1,y_1),\dots,(x_n,y_n)\}.
		\label{eq:edit_dist_identity_matching}
	\end{align}
	The existence of this matching thus leads to an upper bound on the distance 
	\begin{equation*}
		\begin{aligned}
			d_\mathrm{E}(X,Y) &\leq C(\M^*) \\
			&= \sum_{i=1}^n d(x_i, y_i) + 0  + 0 = 0,
		\end{aligned}
	\end{equation*}
	which, combined with $d_{\mathrm{E}}(\seqx,\seqy)\geq0$, implies $d_\mathrm{E}(X,Y)=0$. Conversely, assume that $d_\mathrm{E}(X,Y) = 0$. This implies that both penalisation terms are zero, and therefore every entry of each sequence is included in the matching. Furthermore, this implies the sequences are of equal length, that is, $|\seqx|=|\seqy|$. The only possible \textit{monotone} matching which includes all sequence entries is the $\M^*$ seen in \cref{eq:edit_dist_identity_matching}, thus
	\begin{equation*}
		\begin{aligned}
			 d_{\mathrm{E}}(\seqx,\seqy) = \sum_{i=1}^n d(x_i,y_i) = 0 
		\end{aligned}
	\end{equation*}
	which, since $d(x,y) \geq 0$, implies
	\begin{equation*}
    	\begin{aligned}
    		  d(x_i, y_i) = 0 \quad (\text{for } i = 1,\dots n).
    	\end{aligned}
	\end{equation*}
    since $d$ is a metric it satisfies condition \eqref{metric_condition2:indentity}, and thus we have 
	\begin{equation*}
		\begin{aligned}
			x_i = y_i \quad(\text{for } i=1,\dots,n)
		\end{aligned}
	\end{equation*}
	that is, $X=Y$, confirming that $d_{\mathrm{E}}$ satisfies \eqref{metric_condition2:indentity}. 
	
	The symmetry condition \eqref{metric_condition3:symmetry} follows trivially from the symmetry of $d(x,y)$ and the penalisation terms. 
	
	We now finish with confirming the triangle inequality \eqref{metric_condition4:tri} is satisfied. The approach is almost identical to the proof of \Cref{thm:matching_is_metric} (Part 1), with one key difference: we must ensure all matchings are monotone. Assuming that $X=(x_1,\dots,x_n)$, $Y=(y_1,\dots,y_m)$ and $Z = (z_1,\dots,z_l)$ are three sequences, we seek to show that 
	$$d_\mathrm{E}(X,Y) \leq d_\mathrm{E}(X,Z) + d_\mathrm{E}(Z,Y).$$
	With $\M^*_1$ and $\M_2^*$ denoting optimal monotone matchings for $d_\mathrm{E}(X,Z)$ and $d_\mathrm{E}(Z,Y)$ respectively, that is 
	\begin{align*}
			d_\mathrm{E}(X,Z) = C(\M_1^*) && d_\mathrm{E}(Z,Y) =  C(\M_2^*)
	\end{align*}
	these induce the following matching $\M_3$ of $X$ and $Y$ 
	\begin{align}
		\M_3 = \{(x_i, y_j) \, : \, (x_i, z_k) \in \M^*_1 \text{ and } (z_k, y_j) \in \M^*_2 \text{ for some } z_k \in Z\}
		\label{eq:edit_distance_proof_induced_matching}
	\end{align}
	that is, we match entries of $X$ and $Y$ if they were matched to the same entry of $Z$. 
	
	We now confirm this is a monotone matching. Recall that $\M_3$ is monotone if for any pairs $(x_{i_1}, y_{j_1})$ and $(x_{i_2}, y_{j_2})$ in $\M_3$ we have 
	$$i_1 < i_2 \iff j_1 < j_2.$$
	By the definition of $\M_3$ there exists $z_{k_1}$ and $z_{k_2}$ in $Z$ such that 
	\begin{equation*}
		\begin{aligned}
			(x_{i_1}, z_{k_1}) \in \M_1^* && (z_{k_1}, y_{j_1}) \in \M_2^* \\
			(x_{i_2}, z_{k_2}) \in \M_1^* && (z_{k_2}, y_{j_2}) \in \M_2^*
		\end{aligned}
	\end{equation*}
	Furthermore, since $\M_1^*$ and $\M_2^*$ are monotone we have 
	\begin{equation*}
		\begin{aligned}
			i_1 < i_2 \iff k_1 < k_ 2 && && \text{ and } && && k_1 < k_2 \iff j_1 < j_2
		\end{aligned}
	\end{equation*}
	which therefore implies
	\begin{align*}
		i_1 < i_2 \iff k_1 < k_2 \iff j_1 < j_2
	\end{align*}
	and hence $\M_3$ is also monotone. Now that we have shown the induced matching is indeed monotone, observe that by definition we have the following 
	$$d_{\mathrm{E}}(\seqx,\seqy) \leq C(\M_3)$$
	which implies the triangle inequality will hold if we can show the following inequality is satisfied
	\begin{equation}
		\begin{aligned}
			C(\M_3) \leq d_\mathrm{E}(X,Z) + d_\mathrm{E}(Z,Y).
		\end{aligned}
		\label{eq:edit_distance_proof_ineq}
	\end{equation}
	The argument for this is identical to that used to show \cref{eq:matching_tri_ineq} in the proof of \Cref{thm:matching_is_metric} (Part 1). For brevity, we do not repeat the steps and henceforth assume \cref{eq:edit_distance_proof_ineq} holds. Thus the triangle inequality \eqref{metric_condition4:tri} holds, completing the proof.
\end{proof}
\begin{proof}[Proof of \Cref{thm:edit_distance_metric} (Part 2)]
    To aide this exposition, we write $d_{\mathrm{E},\rho}(\seqx, \seqy)$ in terms of its cost function as follows
    \begin{align*}
		d_{\mathrm{E},\rho}(\seqx,\seqy) = \min_{\M} C(\M)
	\end{align*}
	where $\M$ denotes a monotone matching of $\seqx$ and $\seqy$ and 
	\begin{equation*}
	    \begin{aligned}
	        C(\M) =  \sum_{(\elmtx,\elmty) \in \M} d(\elmtx,\elmty) + \rho( |\setx| + |\sety| - 2|\M| ) 
	    \end{aligned}
	\end{equation*}
	denotes the cost of the matching $\M$.
	
	Conditions \eqref{metric_condition2:indentity} and \eqref{metric_condition3:symmetry} can be shown to hold by following the same reasoning seen in the proof of \Cref{thm:edit_distance_metric} (Part 1). For brevity, we therefore do not repeat the details. We will, however, outline the argument confirming satisfaction of the triangle inequality \eqref{metric_condition4:tri}. 
	
	Assuming $X=(x_1,\dots,x_n)$, $Y=(y_1,\dots,y_m)$ and $Z=(z_1,\dots,z_l)$ are sequences, we are looking to show 
	$$d_{\mathrm{E},\rho}(X,Y) \leq d_{\mathrm{E},\rho}(X,Z) + d_{\mathrm{E},\rho}(Z,Y).$$
	Letting $\M_1^*$ and $\M_2^*$ denote the optimal monotone matchings of $d_{\mathrm{E},\rho}(X,Z)$ and $d_{\mathrm{E},\rho}(Z,Y)$ respectively, that is 
	\begin{align*}
			d_{\mathrm{E},\rho}(X,Z) = C(\M_1^*) && d_{\mathrm{E},\rho}(Z,Y) =  C(\M_2^*)
	\end{align*}
	we can find a monotone matching $\M_3$ of $X$ and $Y$ induced by $\M_1^*$ and $\M_2^*$ as was done in the proof of \Cref{thm:edit_distance_metric} (Part 1), that is, take $\M_3$ as in \cref{eq:edit_distance_proof_induced_matching}, where we matched entries of $X$ and $Y$ if they were matched to the same entry of $Z$.
	
	The next step in proving \eqref{metric_condition4:tri} is to confirm the following inequality holds 
	\begin{equation}
		\begin{aligned}
			C(\M_3)\leq d_{\mathrm{E},\rho}(X,Z) + d_{\mathrm{E},\rho}(Z,Y),
		\end{aligned}
		\label{eq:fp_edit_distance_tri_ineq}
	\end{equation}
    for which we again appeal arguments in a previous proof. Specifically, that of \Cref{thm:fp_matching_is_metric} (Part 2), where a similar inequality was shown to hold for the fixed penalty matching distance, namely that of \cref{eq:fp_matching_tri_ineq}. Recall the argument therein used properties of $\M_3$ to show that every term on the LHS is less than or equal to some unique terms on the RHS. Since the same $\M_3$ has been taken in the present case, the argument can also be applied here. As such, it will henceforth be assumed \cref{eq:fp_edit_distance_tri_ineq} holds. Thus the triangle inequality \eqref{metric_condition4:tri} is satisfied, completing the proof.
\end{proof}

\begin{proof}[Proof of \Cref{thm:dtw_metric_cond_violate}]

    For ease of reference, recall the DTW distance between sequences $\seqx$ and $\seqy$ is given by the following 
    \begin{equation*}
        \begin{aligned}
            d_{\mathrm{DTW}}(\seqx, \seqy) = \min_{\C}\left\{\sum_{(\elmtx,\elmty) \in \C} d(\elmtx, \elmty) \right\}
        \end{aligned}
    \end{equation*}
    where $\C$ is a coupling. 
    
    Observe the symmetry condition \eqref{metric_condition3:symmetry} follows trivially from the symmetry of the ground distance $d(\cdot,\cdot)$, by virtue of it being a metric.
    
    We now show that conditions \eqref{metric_condition2:indentity} and \eqref{metric_condition4:tri} are violated by providing counterexamples. Beginning with \eqref{metric_condition2:indentity}, consider the following two sequences
    \begin{align*}
        \seqx &= (\elmtx_1) &  \seqy &= (\elmty_1, \elmty_2)\\
        &=(\tilde{x}) & &=(\tilde{x},\tilde{x}) 
    \end{align*}
    where $\tilde{\elmtx} \in \X$ denotes an arbitrary element of the underlying space. Clearly, we have $\seqx \not = \seqy$. However, there is only one valid coupling of $\seqx$ and $\seqy$, namely $\C = ((\elmtx_1,\elmty_1),(\elmtx_1,\elmty_2))$. Consequently, the DTW distance is given by
    \begin{equation*}
        \begin{aligned}
            d_{\mathrm{DTW}}(\seqx, \seqy) &= \sum_{(\elmtx,\elmty) \in \C} d(\elmtx, \elmty)\\
            &= d(\elmtx_1,\elmty_1) + d_I(\elmtx_1, \elmty_2)\\
            &= d(\tilde{\elmtx},\tilde{\elmtx}) + d(\tilde{\elmtx}, \tilde{\elmtx}) = 0 + 0
        \end{aligned}
    \end{equation*}
    violating condition \eqref{metric_condition2:indentity}. 
    
    Turning now to condition \eqref{metric_condition4:tri}, consider the following three sequences
    \begin{equation}
        \begin{aligned}
            \seqx &= (\elmtx_1,\elmtx_2) && \seqz &= (\elmtz_1) && \seqy &= (\elmty_1) \\
            &=(\tilde{x},\tilde{x}) && &=(\tilde{x}) && &=(\tilde{y})
        \end{aligned}
        \label{eq:dtw_tri_ineq_counter_ex}
    \end{equation}
    where $\tilde{x},\tilde{y} \in \X$ with $\tilde{x}\not = \tilde{y}$. Now, the only valid coupling of $\seqx$ and $\seqz$ is given by $\C_{\seqx\seqz}=((\elmtx_1,\elmtz_1),(\elmtx_2,\elmtz_1))$, similarly the only coupling of $\seqz$ and $\seqy$ is given by $\C_{\seqz\seqy}=((\elmtz_1,\elmty_1))$, whilst for $\seqx$ and $\seqy$ this will be $\C_{\seqx\seqy}=((\elmtx_1,\elmty_1),(\elmtx_2,\elmty_1))$. This therefore implies 
    \begin{equation*}
        \begin{aligned}
            d_{\mathrm{DTW}}(\seqx,\seqy) &= \sum_{(x,y) \in \C_{\seqx\seqy}} d(x,y) \\
            &= d(\elmtx_1,\elmty_1) + d(\elmtx_2,\elmty_1) \\
            &=d(\tilde{\elmtx},\tilde{\elmty}) + d(\tilde{\elmtx}, \tilde{\elmty})\\
            &=2d(\tilde{\elmtx}, \tilde{\elmty})
        \end{aligned}
    \end{equation*}
    whilst
    \begin{equation*}
        \begin{aligned}
            d_{\mathrm{DTW}}(\seqx,\seqz) + d_{\mathrm{DTW}}(\seqz,\seqy) &= \sum_{(x,z) \in \C_{\seqx\seqz}}d(x,z) + \sum_{(z,y) \in \C_{\seqz\seqy}}d(z, y)\\ 
            &= \left[d(\elmtx_1, \elmtz_1) + d(\elmtx_2,\elmtz_1)\right] + \left[d(\elmtz_1, \elmty_1)\right]\\
            &=\left[d(\tilde{\elmtx}, \tilde{\elmtx}) + d(\tilde{\elmtx}, \tilde{\elmtx})\right] + \left[d(\tilde{\elmtx}, \tilde{\elmty})\right]\\
            &= d(\tilde{\elmtx}, \tilde{\elmty})
        \end{aligned}
    \end{equation*}
    where we have used the fact, since $d$ is a metric, we have $d(\elmtx, \elmtx)=0$. Now, since $\tilde{\elmtx} \not = \tilde{\elmty}$ we have $d(\tilde{\elmtx}, \tilde{\elmty})>0$, implying
    \begin{equation*}
        \begin{aligned}
            d_{\mathrm{DTW}}(\seqx,\seqy) = 2d(\tilde{\elmtx},\tilde{\elmty}) > d(\tilde{\elmtx}, \tilde{\elmty}) = d_{\mathrm{DTW}}(\seqx,\seqz) + d_{\mathrm{DTW}}(\seqz,\seqy)
        \end{aligned}
    \end{equation*}
    and \eqref{metric_condition4:tri} is violated, as desired. This completes the proof.
\end{proof}
\begin{proof}[Proof of \Cref{thm:fp_dtw_metric_conds}]
    For ease of reference, recall the fixed-penalty DTW distance between sequences $\seqx$ and $\seqy$ is given by the following 
    \begin{equation}
        \begin{aligned}
            d_{\mathrm{DTW}, \rho}(\seqx, \seqy) = \min_{\C}\left\{\sum_{(\elmtx,\elmty) \in \C} d(\elmtx, \elmty)  + \rho \cdot w(\C)\right\}
        \end{aligned}
    \end{equation}
    where
    $$w(\C) :=  |\{(\elmtx_i,\elmty_j) \in \C \, : \, (\elmtx_i,\elmty_{j+1}) \in \C \text{ or } (\elmtx_{i+1},\elmty_j) \in \C \}|, $$
    quantifies the amount of warping in $\C$, whilst $\rho > 0$ is a parameter controlling the penalisation incurred for each instance of warping.
    
    We first show that condition \eqref{metric_condition2:indentity} (identity of indiscernibles) holds. To do so, let $\seqx$ and $\seqy$ be two sequences such that $d_{\mathrm{DTW},\rho}(\seqx,\seqy)=0$. Assuming that $\C^*$ denotes an optimal coupling, that is
    \begin{equation*}
        \begin{aligned}
            d_{\mathrm{DTW},\rho}(\seqx,\seqy) = \sum_{(x,y) \in \C^*} d(\elmtx,\elmty) + \rho \cdot w(\C^*)
        \end{aligned}
    \end{equation*}
    implying
    \begin{equation}
        \begin{aligned}
            \sum_{(x,y) \in \C^*} d(\elmtx,\elmty) + \rho \cdot w(\C^*)=0,
        \end{aligned}
        \label{eq:fp_dtw_identity_eq1}
    \end{equation}
    where we here use the fact that $d_{\mathrm{DTW},\rho}(\seqx,\seqy)=0$. Notice since $d(x,y) \geq 0$, $w(\C)\geq0$ and $\rho >0$ this implies each term in \cref{eq:fp_dtw_identity_eq1} must be zero. In particular, we must have $w(\C^*)=0$, that is, no warping has taken place. Thus, each entry of $\seqx$ is paired with exactly one from $\seqy$. Observe this implies $|\seqx|=|\seqy|=N$ and furthermore the only coupling possible is the following 
    $$\C^* = \{(\elmtx_i,\elmty_i) \, : \, i=1,\dots,N\}$$
    that is, we pair the first entries, second entries, and so on. Moreover, due to \cref{eq:fp_dtw_identity_eq1} the sum of pairwise distances in $\C^*$ must be zero, implying
    \begin{equation*}
        \begin{aligned}
            d_{\mathrm{DTW},\rho}(\seqx, \seqy) &= \sum_{(\elmtx,\elmty) \in \C^*}d(\elmtx,\elmty)\\
            &= \sum_{i=1}^N d(\elmtx_i,\elmty_i)=0
        \end{aligned}
    \end{equation*}
    and now since $d$ is a metric we have $d(\elmtx_i, \elmty_i) \geq0$ this implies $d(\elmtx_i,\elmty_i) = 0$ for $i=1,\dots,N$. Finally, using again the fact $d$ is a metric this implies $\elmtx_i=\elmty_i$ for $i=1,\dots,N$ and consequently we have $\seqx=\seqy$. 
    
    Conversely, assume that $\seqx$ and $\seqy$ are sequences such that $\seqx=\seqy$. With $\C^*$ again denoting the coupling obtained by pairing the $i$th entry of $\seqx$ with the $i$th entry of $\seqy$, that is
    $$\C^* = \{(\elmtx_i,\elmty_i) \, : \, i=1,\dots,N\}$$
    where $N=|\seqx|=|\seqy|$, this implies the following upper bound 
    \begin{equation*}
        \begin{aligned}
            d_{\mathrm{DTW},\rho}(\seqx,\seqy) &\leq \sum_{(\elmtx,\elmty) \in \C^*}d(\elmtx,\elmty) + \rho\cdot  w(\C^*) \\
            &= \sum_{i=1}^N d(\elmtx_i,\elmty_i)\\
            &=0
        \end{aligned}
    \end{equation*}
    where in the second line we use the fact $w(\C^*)=0$, whilst the third line follows since $\elmtx_i=\elmty_i$ for $i=1,\dots,N$ by assumption and $d$ is a metric. Finally, since $d_{\mathrm{DTW},\rho}(\seqx,\seqy)\geq0$ by definition, this implies we must have $d_{\mathrm{DTW},\rho}(\seqx,\seqy)=0$. This confirms that condition \eqref{metric_condition2:indentity} holds.
    
    The symmetry condition \eqref{metric_condition3:symmetry} follows trivially from the symmetry of the ground distance $d$ and the penalisation term. 
    
    We now show the triangle inequality \eqref{metric_condition4:tri} is violated. Here as a counterexample we consider the sequence $\seqx$, $\seqy$ and $\seqz$ defined in \cref{eq:dtw_tri_ineq_counter_ex} as seen in the proof of \Cref{thm:dtw_metric_cond_violate}, with $\C_{\seqx\seqz}$, $\C_{\seqz\seqy}$ and $\C_{\seqx\seqy}$ the associated couplings. As in the proof of \Cref{thm:dtw_metric_cond_violate}, this implies 
    \begin{equation*}
        \begin{aligned}
            d_{\mathrm{DTW},\rho}(\seqx,\seqy) &= \sum_{(x,y) \in \C_{\seqx\seqy}} d(x,y) + \rho \cdot w(\C_{\seqx\seqy})\\
            &= d(\elmtx_1,\elmty_1) + d(\elmtx_2,\elmty_1) + \rho \\
            &=d(\tilde{\elmtx},\tilde{\elmty}) + d(\tilde{\elmtx}, \tilde{\elmty}) + \rho\\
            &=2d(\tilde{\elmtx}, \tilde{\elmty}) + \rho
        \end{aligned}
    \end{equation*}
    whilst
    \begin{equation*}
        \begin{aligned}
            d_{\mathrm{DTW},\rho}(\seqx,\seqz) + d_{\mathrm{DTW},\rho}(\seqz,\seqy) &= \left(\sum_{(x,z) \in \C_{\seqx\seqz}}d(x,z)+ \rho \cdot w(\C_{\seqx\seqz})\right) \\
            & \quad + \left(\sum_{(z,y) \in \C_{\seqz\seqy}}d(z, y)+ \rho \cdot w(\C_{\seqz\seqy})\right)\\ 
            &= \left[d(\elmtx_1, \elmtz_1) + d(\elmtx_2,\elmtz_1)\right] + \left[d(\elmtz_1, \elmty_1)\right]\\
            &=\left[d(\tilde{\elmtx}, \tilde{\elmtx}) + d(\tilde{\elmtx}, \tilde{\elmtx}) + \rho\right] + \left[d(\tilde{\elmtx}, \tilde{\elmty})\right]\\
            &= d(\tilde{\elmtx}, \tilde{\elmty}) + \rho
        \end{aligned}
    \end{equation*}
    where we have used the fact, since $d$ is a metric, we have $d(\elmtx, \elmtx)=0$. Now, since $\tilde{\elmtx} \not = \tilde{\elmty}$ we have $d(\tilde{\elmtx}, \tilde{\elmty})>0$, implying
    \begin{equation*}
        \begin{aligned}
            d_{\mathrm{DTW},\rho}(\seqx,\seqy) = 2d(\tilde{\elmtx},\tilde{\elmty}) + \rho > d(\tilde{\elmtx}, \tilde{\elmty}) + \rho = d_{\mathrm{DTW},\rho}(\seqx,\seqz) + d_{\mathrm{DTW},\rho}(\seqz,\seqy)
        \end{aligned}
    \end{equation*}
    thus violating \eqref{metric_condition4:tri}, as desired. This completes the proof.
\end{proof}

\section{Path distances}

\label{sec:dist_interactions}

In this section, we provide further details regarding two path distances, one of which was invoked in the data analysis of \Cref{sec:data_analysis}. Both are defined via maximally-sized substructures shared by the two paths being compared, considering in particular common subsequences and subpaths (\Cref{fig:path_distances}).

For a path $\elmtx = (x_1,\dots, x_n)$, where each $x_i \in \V$ with $\V$ some vertex set, we denote a \textit{subpath} of $x$ from index $i$ to $j$ by the following
$$x_{i:j} = (x_i, \dots, x_j) $$
where $1 \leq i \leq j \leq n$ (\Cref{fig:path_distances_subpath}).
More generally, assume that $\bm{v} = (v_1, \dots, v_s)$ with $1 \leq v_1 < v_2 < \dots < v_s \leq n$, then a \textit{subsequence} of $x$ can be obtained by indexing with $\bm{v}$ as follows
$$x_{\bm{v}} = (x_{v_1}, \dots, x_{v_s})$$
which will be of length $s$ (\Cref{fig:path_distances_subseq}). Observe that every subpath of $x$ is also a subsequence, making a subsequence the more general of the two structures. 

Given another path $y=(y_1,\dots,y_m)$, one also can consider the notion of \textit{common} subpaths and subsequences. Namely, a common subpath of $x$ and $y$ occurs when we have
$$x_{i:j}=y_{l:k}$$
for some $1 \leq i \leq j \leq n$ and $1 \leq l \leq k \leq m$. Similarly, a common subsequece of $x$ and $y$ occurs when 
$$x_{\bm{v}}=y_{\bm{u}}$$
for some $1 \leq v_1 < v_2 < \dots < v_s \leq n$ and $1 \leq u_1 < u_2 < \dots < u_s \leq m$. 

The more similar $x$ and $y$ are the larger we might expect their common subpaths or subsequences to be. Following this rationale, one can define distances between $x$ and $y$ by finding common subpaths or subsquences which \textit{maximumal}, that is, one for which there exist no other common subpaths or subsequences of greater length. This leads to the following definitions.

\begin{definition}[Longest common subsequence distance]
    For two paths $x=(x_1,\dots,x_n)$ and $y=(y_1,\dots,y_m)$ the longest common subsequence (LCS) distance is given by the following 
    \begin{equation*}
        d_\mathrm{LCS}(x, y) := n + m - 2\delta_\mathrm{LCS}
    \end{equation*}
    where 
    \begin{equation*}
         \delta_\mathrm{LCS} = \max\{|\bm{v}| = |\bm{u}|  \, : \, x_{\bm{v}}=y_{\bm{u}} \},
    \end{equation*}
    where $|\bm{v}|$ denotes the length of $\bm{v}$, so that $\delta_{\mathrm{LSP}}$ denotes the maximum length subsequence common to both $x$ and $y$.
    \label{def:lcs}
\end{definition}

\begin{definition}[Longest common subpath distance]
    For two paths $x=(x_1,\dots,x_n)$ and $y=(y_1,\dots,y_m)$ the longest common subpath (LSP) distance is given by the following 
    $$d_\mathrm{LSP}(x,y) := n + m - \delta_\mathrm{LSP}$$ 
    where 
    $$\delta_\mathrm{LSP} = \max\{|i:j| = |l:k|  \, : \, x_{i:j}=y_{l:k}\}, $$ 
    denoting the maximum length subpath common to both $x$ and $y$.
    \label{def:lsp}
\end{definition}

Both the LCS and LSP distances can be shown to be metrics, that is, they satisfy all three metric conditions, as we summarise via the following result.

\begin{proposition}
    Both $d_{\mathrm{LCS}}$ and $d_{\mathrm{LSP}}$ staisfy metric conditions \eqref{metric_condition2:indentity}-\eqref{metric_condition4:tri}.
    \label{thm:lcs_lsp_metrics}
\end{proposition}

\begin{proof}[Proof of \Cref{thm:lcs_lsp_metrics}]

Consider first condition \eqref{metric_condition2:indentity} (identity of indiscernibles). If we have two paths $x=(x_1,\dots,x_n)$ and $y=(y_1,\dots,y_m)$ such that $d_{\mathrm{LCS}}(x,y)=0$, this implies 
\begin{equation}
     \begin{aligned}
     n+m-2\delta_{\mathrm{LCS}}=0.
     \end{aligned}
     \label{eq:path_metric_proof1}
\end{equation}
Moreover, by definition $\delta_\mathrm{LCS} \leq \min(n,m)$, since a common subsequence cannot be longer than the shorter path. We claim that this implies $n=m$. Towards doing so, if we assume $n<m$ this implies 
$$n+m > 2n \geq 2\delta_\mathrm{LCS}$$
where we have used the face $\delta_\mathrm{LCS} \leq n$, which contradicts \cref{eq:path_metric_proof1}. A similar contradiction can be made if we assume $n > m$, and consequently we must have $n=m$. Substituting this into \cref{eq:path_metric_proof1} leads to
$\delta_\mathrm{LCS}=n=m$
which implies that $x$ and $y$ share a common subsequence of the same length as themselves, that is $x=y$. Conversely, if $x=y$ then it should be clear that the maximum common subsequence will be the one including all their entries, that is $\delta_\mathrm{LCS}=n=m$ and hence $$d_\mathrm{LCS}(x,y)=n+m-2\delta_\mathrm{LCS}=0.$$ 
This proves that \eqref{metric_condition2:indentity} holds for the LCS distance, and an identical argument can be used to show it similarly holds for the LSP distance.

The symmetry condition \eqref{metric_condition3:symmetry} for both the LCS and LSP distances follows trivially from the symmetry in the definition of common subsequences and subpaths, respectively.

Finally we turn to the triangle inequality \eqref{metric_condition4:tri}, considering first the LCS distance. Assume that $x=(x_1,\dots,x_n)$, $y=(y_1,\dots,y_m)$ and $z=(z_1,\dots,z_k)$ are three paths and that $\delta_{xz}$, $\delta_{zy}$ and $\delta_{xy}$ are such that 
\begin{align*}
    d_\mathrm{LCS}(x,y)=n+m-2\delta_{xy} && d_\mathrm{LCS}(x,z)=n+k-2\delta_{xz} && d_\mathrm{LCS}(z,y)=n+m-2\delta_{zy}, 
\end{align*}
then, assuming the triangle inequality holds, we have
$$d_\mathrm{LCS}(x,y) \leq  d_\mathrm{LCS}(x,z) + d_\mathrm{LCS}(z,y)$$
which is equivalent to the following 
$$n+m-2\delta_{xy} \leq n+k-2\delta_{xz} + n+m-2\delta_{zy}$$
which is true if and only if 
\begin{equation}
    \begin{aligned}
    \delta_{xz} + \delta_{zy} - k \leq \delta_{xy}.
    \end{aligned}
    \label{eq:lcs_lsp_proof_tri_ineq}
\end{equation}
Notice that if \cref{eq:lcs_lsp_proof_tri_ineq} holds then one can trace the implications back to conclude the triangle inequality also holds. Towards doing so, we consider a finding the common subsequence between $x$ and $y$ induced by those between $x$ and $z$ and $y$ and $z$, which will allow us to obtain the desired lower bound. 

To aide this exposition we introduce some notation. In particular, for two subsequences $\bm{v}$ and $\bm{u}$ of $[s]=(1,\dots,s)$ we can extend the notion of unions and intersections used for sets, that is $\bm{v} \cup \bm{u}$ and $\bm{v} \cap \bm{u}$ respectively, where if $\bm{w}=\bm{v} \cap \bm{u}$ then each entry thereof $w_i$ appears in both $\bm{v}$ and $\bm{u}$, whilst if $\bm{w} = \bm{v} \cup \bm{u}$ then each $w_i$ appears in at least one of $\bm{u}$ and $\bm{v}$. Moreover, with $|\bm{v}|$ denoting the length of subsequence $\bm{v}$, as for sets the following identity will hold 
$$|\bm{v}| + |\bm{u}| - |\bm{v} \cap \bm{u}|=|\bm{v} \cup \bm{u}|.$$
Now suppose that $\bm{v}_{xz}$ and $\bm{v}_{zy}$ index a maximal common subsequences of $z$ with $x$ and $y$ respectively, that is, both are subsequences of $[k]$ such that
\begin{align*}
     x_{\bm{v}_{xz}}=z_{\bm{u}_{xz}} && z_{\bm{v}_{zy}}=y_{\bm{u}_{zy}} 
\end{align*}
for some subsequences $\bm{v}_{xz}$ of $[n]$ and $\bm{u}_{zy}$ of $[m]$, where $|\bm{u}_{xz}|=\delta_{xz}$ and $|\bm{v}_{xz}|=\delta_{zy}$. Notice now the entries of $z$ indexed by $\bm{u}_{xz} \cap \bm{v}_{zy}$ will induce a common subsequence of $x$ and $y$ (by considering the associated entries of each). As such, if we let $\delta^*=|\bm{u}_{xz} \cap \bm{v}_{zy}|$ we will have 
$$\delta^* \leq \delta_{xy}$$
by virtue of $\delta_{xy}$ being the \textit{maximal} length of a common subsequence between $x$ and $y$. Also by the inclusion-exclusion-like identity above we will have
\begin{equation*}
    \begin{aligned}
    \delta_{xz} + \delta_{zy} - \delta^* = |\bm{u}_{xz} \cup \bm{v}_{zy}| \leq k
    \end{aligned}
\end{equation*}
where the inequality here follows since $\bm{u}_{zx}$ and $\bm{v}_{zy}$ are subsequences of $[k]$ (since they index $z$, which is of length $k$). Combining these last two inequalities thus leads to the following 
\begin{equation*}
    \begin{aligned}
    \delta_{xz} + \delta_{zy} - k \leq \delta^* \leq \delta_{xy},
    \end{aligned}
\end{equation*}
hence confirming \cref{eq:lcs_lsp_proof_tri_ineq} holds and proving the triangle inequality holds for the LCS distance. 

A similar argument can be used to prove that condition \eqref{metric_condition4:tri} also holds for the LSP distance. For brevity we will not give full exposition here, but we note the only key difference will be the notion of intersections and unions. If we introduce the shorthand notation $(i:j)=(i,\dots,j)$ where $1 \leq i\leq j \leq s$, denoting the subpath of $[s]$ from $i$ to $j$ (notice this is consistent with notation used in \Cref{def:lsp}), then we naturally have the following
\begin{align*}
    (i:j) \cap (l:k) = (\max(i,l):\min(j,k)) && (i:j) \cup (l:k) = (\min(i,l):\max(j,k)),
\end{align*}
and moreover if $|(i:j)| = j-i+1$ denotes subpath length we will similarly have the following identity 
$$|(i:j)| + |(l:k)| - |(i:j) \cap (l:k)|=|(i:j) \cup (l:k)|.$$
With these results one can follow the rationale used for the LCS distance, replacing subsequences with subapaths, to show that the triangle inequality is satisfied. 

Thus conditions \eqref{metric_condition2:indentity}, \eqref{metric_condition3:symmetry} and \eqref{metric_condition4:tri} all hold for both the LCS and LSP distances, completing the proof. 

\end{proof}

Finally, we discuss computation. Both of these distances can be computed via dynamic programming, much like the edit and DTW distances (\Cref{sec:computing_edit_distance,sec:computing_dtw_distance}), with a time complexity of $\mathcal{O}(nm)$. Infact, the LCS distance can be seen as an instance of the fixed-penalty edit distance $d_{\mathrm{E},\rho}$ with ground distance given by 
$$d(x_i,y_j) = \begin{cases}
0 & \text{if } x_i=y_j \\
2 & \text{otherwise}
\end{cases}$$
and $\rho=1$. Consequently, one can apply \Cref{alg:fp_edit_distance} or \Cref{alg:fp_edit_distance_light} directly, substituting in these values for $d(\cdot,\cdot)$ and $\rho$.

% \begin{equation*}
%     d_\mathrm{LCS}(x_{1:i},y_{1:j}) =\min\begin{cases}
%     d_{\mathrm{LCS}}(x_{1:(i-1)}, y_{1:j}) + 1\\
%     d_{\mathrm{LCS}}(x_{1:i}, y_{1:(j-1)}) + 1 \\
%     d_{\mathrm{LCS}}(x_{1:(i-1)}, y_{1:(j-1)}) + d(x_i,y_j).
%     \end{cases}
% \end{equation*}

The approach to evaluate $d_\mathrm{LSP}$ is slightly different. In this case, we essentially scan over $x$ and $y$ and keep track of the common subpaths seen. Formally, we construct an $n \times m$ matrix $Q$ incrementally via the following recursive formula
$$Q_{(i+1)(j+1)} = \begin{cases}
Q_{ij} + 1 & \text{if } x_i=y_j \\
0 & \text{otherwise}
\end{cases},
$$
where when common subpaths appear between $x$ and $y$ one will see increments in $Q$ diagonally. The maximum length of a subpath can thus be obtained by taking the element-wise maximum of $Q$, that is $\delta_\mathrm{LSP} = \max_{ij} Q_{ij}$, which can then be plugged into \Cref{def:lsp} to compute $d_\mathrm{LSP}(x,y)$. We summarise this in \Cref{alg:lsp_distance}, where we keep track of the maximum in $Q$ as it is filled. Moreover, a lighter-memory algorithm is outlined in \Cref{alg:lsp_distance_light}, making use of the fact we only need to know the current and previous rows of $Q$.

\label{sec:path_distances}

\section{Pseudocode}

\begin{algorithm}[H]
	\SetAlgoLined
	\KwData{Sequences $X = (x_1,\dots,x_n)$ and $Y=(y_1,\dots,y_m)$}
	\KwIn{$d(\cdot, \cdot)$ a distance metric between sequence entries}
	\KwResult{$d_{\mathrm{E}}(X,Y)$ (\Cref{def:edit})}
	$C \in \mathbb{R}^{(n+1)\times (m+1)}$\;
	$C_{11}=0$\;
	$C_{(i+1)1} = C_{i1} + d(x_i,\Lambda)$ (for $i = 1,\dots,n$)\;
	$C_{1(j+1)} = C_{1j} + d(y_j,\Lambda)$ (for $j = 1,\dots,m$)\;
	
	\For{$i = 1,\dots, n$}{
		\For{$j = 1,\dots, m$}{
			$C_{(i+1)(j+1)} = \min\begin{cases}
			C_{ij} + d(x_i,y_j)\\
			C_{(i+1)j} + d(y_j,\Lambda) \\
			C_{i(j+1)} + d(x_i,\Lambda)
			\end{cases}$
		}
	}
	\Return{$C_{(n+1)(m+1)}$}
	\vspace{0.2cm}
	\caption{Evaluating edit distance $d_{\mathrm{E}}$}
	\label{alg:edit_distance}
\end{algorithm}
\begin{algorithm}[H]
	\SetAlgoLined
	\KwData{Sequences $X = (x_1,\dots,x_n)$ and $Y=(y_1,\dots,y_m)$}
	\KwIn{$d(\cdot, \cdot)$ a distance metric between sequence entries}
	\KwResult{$d_{\mathrm{E}}(X,Y)$  (\Cref{def:edit})}
	$Z^{\text{prev}}, Z^{\text{curr}} \in \mathbb{R}^{(m+1)}$\;
	$Z^{\mathrm{prev}}_1=0$, $Z^{\mathrm{curr}}_1=0$\;
	$Z^{\text{prev}}_{i+1} = Z^{\text{prev}}_i + d(y_i,\Lambda)$ (for $i = 1,\dots,m$)\;
	\For{$i = 1,\dots, n$}{
		$Z^\text{curr}_1 = Z^\text{curr}_1 + d(x_i, \Lambda)$\;
		\For{$j = 1,\dots, m$}{
			$Z^\text{curr}_{j+1} = \min\begin{cases}
			Z^\text{prev}_{j} + d(x_i,y_j) \\ Z^\text{prev}_{j+1} + d(x_i,\Lambda) \\ Z^\text{curr}_{j} + d(y_j,\Lambda)
			\end{cases}$
		}
		$Z^\mathrm{prev} = Z^\text{curr}$
	}
	\Return{$Z^\mathrm{curr}_{m+1}$}
	\vspace{0.2cm}
	\caption{Evaluating edit distance $d_\mathrm{E}$ (light memory)}
	\label{alg:edit_distance_light}
\end{algorithm}

\newpage
\begin{algorithm}[H]
	\SetAlgoLined
	\KwData{Sequences $X = (x_1,\dots,x_n)$ and $Y=(y_1,\dots,y_m)$}
	\KwIn{$d(\cdot, \cdot)$ a distance metric between sequence entries}
	\KwResult{$d_{\mathrm{E},\rho}(X,Y)$ (\Cref{def:fp_edit})}
	$C \in \mathbb{R}^{(n+1)\times (m+1)}$\;
	$C_{(i+1)1} = i \rho$ (for $i = 0,\dots,n$)\;
	$C_{1(j+1)} = j \rho $ (for $j = 0,\dots,m$)\;
	
	\For{$i = 1,\dots, n$}{
		\For{$j = 1,\dots, m$}{
			$C_{(i+1)(j+1)} = \min\begin{cases}
			C_{ij} + d(x_i,y_j)\\
			C_{i(j+1)} + \rho \\
			C_{(i+1)j} + \rho
			\end{cases}$
		}
	}
	\Return{$C_{(n+1)(m+1)}$}
	\vspace{0.2cm}
	\caption{Evaluating fixed-penalty edit distance $d_{\mathrm{E},\rho}$}
	\label{alg:fp_edit_distance}
\end{algorithm}

\begin{algorithm}[H]
	\SetAlgoLined
	\KwData{Sequences $X = (x_1,\dots,x_n)$ and $Y=(y_1,\dots,y_m)$}
	\KwIn{$d(\cdot, \cdot)$ a distance metric between sequence entries}
	\KwResult{$d_{\mathrm{E},\rho}(X,Y)$ (\Cref{def:fp_edit})}
	$Z^{\text{prev}}, Z^{\text{curr}} \in \mathbb{R}^{(m+1)}$\;
	$Z^{\mathrm{prev}}_1=0$, $Z^{\mathrm{curr}}_1=0$\;
	$Z^{\text{prev}}_{i+1} = Z^{\text{prev}}_i + \rho$ (for $i = 1,\dots,m$)\;
	\For{$i = 1,\dots, n$}{
		$Z^\text{curr}_1 = Z^\text{curr}_1 + d(x_i, \Lambda)$\;
		\For{$j = 1,\dots, m$}{
			$Z^\text{curr}_{j+1} = \min\begin{cases}
			Z^\text{prev}_{j} + d(x_i,y_j) \\ Z^\text{prev}_{j+1} + \rho \\ Z^\text{curr}_{j} + \rho
			\end{cases}$
		}
		$Z^\mathrm{prev} = Z^\text{curr}$
	}
	\Return{$Z^\mathrm{curr}_{m+1}$}
	\vspace{0.2cm}
	\caption{Evaluating fixed-penalty edit distance $d_{\mathrm{E},\rho}$ (light memory)}
	\label{alg:fp_edit_distance_light}
	
\end{algorithm}

\newpage

\begin{algorithm}[H]
	\SetAlgoLined
	\KwData{Sequences $X = (x_1,\dots,x_n)$ and $Y=(y_1,\dots,y_m)$}
	\KwIn{$d(\cdot, \cdot)$ a distance metric between sequence entries}
	\KwResult{$d_{\mathrm{DTW}}(X,Y)$ (\Cref{def:dtw})}
	$C \in \mathbb{R}^{(n+1)\times (m+1)}$\;
	$C_{11}=0$\;
	$C_{(i+1)1} =\infty$ (for $i = 1,\dots,n$)\;
	$C_{1(j+1)} =\infty$ (for $j = 1,\dots,m$)\;
	
	\For{$i = 1,\dots, n$}{
		\For{$j = 1,\dots, m$}{
			$C_{(i+1)(j+1)} = d(\elmtx_i,\elmty_j) + \min\{
			C_{ij},C_{(i+1)j},C_{i(j+1)}\}$
			}
	}
	\Return{$C_{(n+1)(m+1)}$}
	\vspace{0.2cm}
	\caption{Evaluating dynamic time warping distance $d_{\mathrm{DTW}}$}
	\label{alg:dtw_distance}
\end{algorithm}

\begin{algorithm}[H]
	\SetAlgoLined
	\KwData{Sequences $X = (x_1,\dots,x_n)$ and $Y=(y_1,\dots,y_m)$}
	\KwIn{$d(\cdot, \cdot)$ a distance metric between sequence entries}
	\KwResult{$d_{\mathrm{DTW}}(X,Y)$ (\Cref{def:dtw})}
	$Z^{\text{prev}}, Z^{\text{curr}} \in \mathbb{R}^{(m+1)}$\;
	$Z^{\mathrm{prev}}_1=0$, $Z^{\mathrm{curr}}=\infty$\;
	$Z^{\text{prev}}_{i+1} = \infty$ (for $i = 1,\dots,m$)\;
	\For{$i = 1,\dots, n$}{
		\For{$j = 1,\dots, m$}{
			$Z^\text{curr}_{j+1} = d(\elmtx_i,\elmty_j) + \min\{Z^\mathrm{prev}_j, Z^{\mathrm{curr}}_j, Z^{\mathrm{prev}}_{j+1}\}$
		}
		$Z^\mathrm{prev} = Z^\text{curr}$
	}
	\Return{$Z^\mathrm{curr}_{m+1}$}
	\vspace{0.2cm}
	\caption{Evaluating dynamic time warping distance $d_{\mathrm{DTW}}$ (light memory)}
	\label{alg:dtw_distance_light}
\end{algorithm}

\newpage

\newpage
\begin{algorithm}[H]
	\SetAlgoLined
	\KwData{Sequences $X = (x_1,\dots,x_n)$ and $Y=(y_1,\dots,y_m)$}
	\KwIn{$d(\cdot, \cdot)$ a distance metric between sequence entries}
	\KwResult{$d_{\mathrm{DTW},\rho}(X,Y)$ (\Cref{def:fp_dtw})}
	$C \in \mathbb{R}^{(n+1)\times (m+1)}$\;
	$C_{11}=0$\;
	$C_{(i+1)1} =\infty$ (for $i = 1,\dots,n$)\;
	$C_{1(j+1)} =\infty$ (for $j = 1,\dots,m$)\;
	
	\For{$i = 1,\dots, n$}{
		\For{$j = 1,\dots, m$}{
			$C_{(i+1)(j+1)} = d(\elmtx_i,\elmty_j) + \min\{
			C_{ij},C_{(i+1)j}+\rho,C_{i(j+1)}+\rho\}$
			}
	}
	\Return{$C_{(n+1)(m+1)}$}
	\vspace{0.2cm}
	\caption{Evaluating fixed-penalty dynamic time warping distance $d_{\mathrm{DTW},\rho}$}
	\label{alg:fp_dtw_distance}
\end{algorithm}

\begin{algorithm}[H]
	\SetAlgoLined
	\KwData{Sequences $X = (x_1,\dots,x_n)$ and $Y=(y_1,\dots,y_m)$}
	\KwIn{$d(\cdot, \cdot)$ a distance metric between sequence entries}
	\KwResult{$d_{\mathrm{DTW},\rho}(X,Y)$ (\Cref{def:fp_dtw})}
	$Z^{\text{prev}}, Z^{\text{curr}} \in \mathbb{R}^{(m+1)}$\;
	$Z^{\mathrm{prev}}_1=0$, $Z^{\mathrm{curr}}_1=\infty$\;
	$Z^{\text{prev}}_{i+1} = \infty$ (for $i = 1,\dots,m$)\;
	\For{$i = 1,\dots, n$}{
		\For{$j = 1,\dots, m$}{
			$Z^\text{curr}_{j+1} = d(\elmtx_i,\elmty_j) + \min\{Z^\mathrm{prev}_j, Z^{\mathrm{curr}}_j+\rho, Z^{\mathrm{prev}}_{j+1}+\rho\}$
		}
		$Z^\mathrm{prev} = Z^\text{curr}$
	}
	\Return{$Z^\mathrm{curr}_{m+1}$}
	\vspace{0.2cm}
	\caption{Evaluating fixed-penalty dynamic time warping distance $d_{\mathrm{DTW},\rho}$ (light memory)}
	\label{alg:fp_dtw_distance_light}
\end{algorithm}

\newpage 

\begin{algorithm}[H]
	\SetAlgoLined
	\KwData{Paths $x = (x_1,\dots,x_n)$ and $y=(y_1,\dots,y_m)$}	\KwResult{$d_{\mathrm{LSP}}(x,y)$ (\Cref{def:lsp})}
	$Q \in \mathbb{\mathbb{Z}_+}^{(n+1)\times (m+1)}$\;
	$Q_{11}=0$\;
	$Q_{(i+1)1} = 0$ (for $i = 1,\dots,n$)\;
	$Q_{1(j+1)} = 0$ (for $j = 1,\dots,m$)\;
	$\delta = 0$\;
	\For{$i = 1,\dots, n$}{
		\For{$j = 1,\dots, m$}{
		    \uIf{$x_i=y_j$}{
		        $Q_{(i+1)(j+1)} = Q_{ij} + 1$
		        
		        $\delta = \max\left(z, Q_{(i+1)(j+1)}\right)$
		    }
			\Else{
			    $Q_{(i+1)(j+1)} = 0$
			}
		}
	}
	\Return{$n + m - 2\delta$}
	\caption{Evaluating LSP distance $d_{\mathrm{LSP}}$}
	\label{alg:lsp_distance}
\end{algorithm}

\begin{algorithm}[H]
	\SetAlgoLined
	\KwData{Paths $x = (x_1,\dots,x_n)$ and $y=(y_1,\dots,y_m)$}	\KwResult{$d_{\mathrm{LSP}}(x,y)$ (\Cref{def:lsp})}
	$Z^\mathrm{prev}, Z^\mathrm{curr} \in \mathbb{Z}_+^{(m+1)}$\;
	$Z^\mathrm{prev}_{i+1}=Z^\mathrm{curr}_{i+1}=0$ (for $i=0,\dots,m$)\;
	$\delta = 0$\;
	\For{$i = 1,\dots, n$}{
		\For{$j = 1,\dots, m$}{
		    \uIf{$x_i=y_j$}{
		        $Z^\mathrm{curr}_{j+1} = Z^\mathrm{prev}_j + 1$
		        
		        $\delta = \max\left(z, Z^\mathrm{curr}_{j+1}\right)$
		    }
			\Else{
			    $Z^\mathrm{curr}_{j+1} = 0$
			}
		}
		$Z^\mathrm{prev}=Z^\mathrm{curr}$
	}
	\Return{$n + m - 2\delta$}
	\caption{Evaluating LSP distance $d_{\mathrm{LSP}}$ (light memory)}
	\label{alg:lsp_distance_light}
\end{algorithm}

% \begin{algorithm}[H]
% 	\SetAlgoLined
% 	\KwData{Sequences $X = (x_1,\dots,x_n)$ and $Y=(y_1,\dots,y_m)$}
% 	\KwResult{$d_\mathrm{LCS}(X,Y)$}
% 	$Z^{\text{prev}}, Z^{\text{curr}} \in \mathbb{Z}^{(m+1)}$\;
% 	$Z^{\text{prev}}_{i} = (i-1)$ (for $i = 1,\dots,m+1$)\;
% 	\For{$i = 1,\dots, n$}{
% 		$Z^\text{curr}_1 = i$\;
% 		\For{$j = 1,\dots, m$}{
% 			$Z^\text{curr}_{j+1} = \min\{Z^\text{prev}_j + 2 \cdot \mathbf{1}(x_i\not =y_i), Z^\text{prev}_{j+1} + 1, Z^\text{curr}_{j} + 1\} $
% 		}
% 		$Z^\text{prev} = Z^\text{curr}$
% 	}
% 	\Return{$Z^{\mathrm{curr}}_{m+1}$}
% 	\vspace{0.2cm}
% 	\caption{Evaluating LCS distance (light memory).}
% \end{algorithm}

% \begin{algorithm}[H]
% 	\SetAlgoLined
% 	\caption{Evaluating LCS distance.}
% 	\KwData{Sequences $X = (x_1,\dots,x_n)$ and $Y=(y_1,\dots,y_m)$}
% 	\KwResult{$d_\mathrm{LCS}(X,Y)$}
% 	$C \in \mathbb{Z}^{(n+1)\times (m+1)}$\;
% 	$C_{i1} = (i-1)$, for $i = 1,\dots,n+1$\;
% 	$C_{1j} = (j-1)$, for $j = 1,\dots,m+1$\;
	
% 	\For{$i = 2,\dots, n+1$}{
% 		\For{$j = 2,\dots, m+1$}{
% 			$C_{ij} = \min\{C_{(i-1)(j-1)} + 2 \cdot \mathbf{1}(x_i\not =y_i), C_{(i-1)j} + 1, C_{i(j-1)} + 1\} $
% 		}
% 	}
% 	\Return{$C_{(n+1)(m+1)}$}
% 	\vspace{0.2cm}
% \end{algorithm}

\end{appendix}

\end{document}